%% file: 0-offline-fair.tex
\documentclass[runningheads]{llncs}
\usepackage{booktabs} 
\newcommand{\ifm}{\ensuremath{\operatorname{IFM}}\xspace}
\newcommand{\gfm}{\ensuremath{\operatorname{GFM}}\xspace}

\newcommand{\vw}{\ensuremath{\operatorname{VOM}}\xspace}

\newcommand{\sbo}{\ensuremath{\operatorname{SAMP-B}}\xspace}
\newcommand{\sab}{\ensuremath{\operatorname{SAMP-AB}}\xspace}

\newcommand{\alggreedy}{\ensuremath{\mathsf{GREEDY}}\xspace}
\newcommand{\algranking}{\ensuremath{\mathsf{RANKING}}\xspace}
\newcommand{\alglu}{\ensuremath{\mathsf{JL}}\xspace}
\newcommand{\algbru}{\ensuremath{\mathsf{BSSX}}\xspace}
\newcommand{\algmgs}{\ensuremath{\mathsf{MGS}}\xspace}

\newcommand{\cri}{\ensuremath{\mathsf{CR1}}\xspace}
\newcommand{\crii}{\ensuremath{\mathsf{CR2}}\xspace}

\usepackage{macros_pan}
\usepackage{subcaption}
\usepackage{multirow}
\usepackage[normalem]{ulem}
\usepackage[numbers,sort&compress]{natbib}

\begin{document}
\title{Fairness Maximization among \\ Offline Agents in Online-Matching Markets}
%
%
\author{Will Ma\inst{1} \and
Pan Xu\inst{2} \and
Yifan Xu\inst{3}}
%
\institute{Columbia University \and New Jersey Institute of Technology
 \and Southeast University\\
\email{willma353@gmail.com, pxu@njit.edu, xyf@seu.edu.cn}}

%
\maketitle              

\begin{abstract}
Matching markets involve heterogeneous agents (typically from two parties) who are paired for mutual benefit.  During the last decade, matching markets have emerged and grown rapidly through the medium of the Internet. They have evolved into a new format, called \emph{Online Matching Markets} (OMMs), with examples ranging from crowdsourcing to online recommendations to ridesharing. There are two features distinguishing OMMs from traditional matching markets. One is the dynamic arrival of one side of the market: we refer to these as \emph{online agents} while the rest are \emph{offline agents}. Examples of online and offline agents include keywords (online) and sponsors (offline) in Google Advertising; workers (online) and tasks (offline) in Amazon Mechanical Turk (AMT); riders (online) and drivers (offline when restricted to a short time window) in ridesharing. The second distinguishing feature of OMMs is the real-time decision-making element.

However, studies have shown that the algorithms making decisions in these OMMs leave disparities in the match rates of offline agents.
For example, tasks in neighborhoods of low socioeconomic status rarely get matched to gig workers, and drivers of certain races/genders get discriminated against in matchmaking.
In this paper, we propose online matching algorithms which optimize for either
individual or group-level fairness among offline agents in OMMs. We present two linear-programming (LP) based sampling algorithms, which achieve online competitive ratios at least $0.725$ for individual fairness maximization (IFM) and $0.719$ for group fairness maximization (GFM), respectively. There are two key ideas helping us break the barrier of $1-1/\sfe$. One is \emph{boosting}, which is to adaptively re-distribute all sampling probabilities only among those offline available neighbors for every arriving online agent. The other is \emph{attenuation}, which aims to balance the matching probabilities among offline agents with different mass allocated by the benchmark LP.  We conduct extensive numerical experiments and results show that our boosted version of sampling algorithms are not only conceptually easy to implement but also highly effective in practical instances of fairness-maximization-related models. 
\end{abstract}

\section{Introduction} 

Matching markets involve heterogeneous agents (typically from two parties) who are paired for mutual benefit.  During the last decade, matching markets have emerged and grown rapidly through the medium of the Internet. They have evolved into a new format, called \emph{Online Matching Markets} (OMMs), with examples ranging from crowdsourcing markets to online recommendations to rideshare. There are two features distinguishing OMMs from traditional matching markets. The first is the dynamic arrivals of a part of the agents in the system, which are referred to as \emph{online agents}, \eg keywords in Google Advertising, workers in Amazon Mechanical Turk (AMT), and riders in Uber/Lyft. As opposed to online agents, the other part is \emph{offline agents}, such as sponsors in Google advertising, tasks in AMT, and drivers in rideshare platforms, whose information is typically known as \emph{a priori}. The second feature is the instant match-decision requirement. It is highly desirable to match each online agent with
one (or multiple) offline agent(s) upon its arrival, due to the low ``patience'' of online agents. There is a large body of research that studies matching-policy design and/or pricing mechanism in OMMs~\cite{bimpikis2019spatial,ma2018spatio}. \emph{The main focus of this paper, instead, is fairness among offline agents in OMMs}. Consider the following two motivating examples.

 \xhdr{Fairness among task requesters in mobile crowdsourcing markets (MCM)}.  In MCMs like Gigwalk and TaskRabbit, each offline task has specific location information, and online workers have to travel to that location to complete it (\eg cleaning one's house, delivering some package). It has been reported that online workers selectively choose offline tasks based on their biased perceptions.  For example, survey results on TaskRabbit in Chicago~\cite{thebault2015avoiding} showed that  workers try to avoid tasks that locate in areas with low economic status like the South Side of Chicago. This causes a much lower completion rate for tasks requested by users of low socioeconomic status and as a result, ``low socioeconomic-status  areas are currently less able to take advantage of the benefit of mobile crowdsourcing markets.''

  \xhdr{Fairness among drivers in ride-hailing services}.  There are several reports showing the earning gap among drivers based on their demographic factors such as age, gender and race, see, \eg~\cite{cook2018,rosenblat2016}. In particular, \cite{wage-gap} has reported that ``Black Uber and Lyft drivers earned \$13.96 an hour compared to the \$16.08 average for all other drivers'' and ``Women drivers reported earning an average of \$14.26 per hour, compared to \$16.61 for men''. The wage gap among drivers from different demographic groups is partially due to the biased perceptions and resulting selective choices from riders who cancel matches with drivers from these groups.

In this paper, we study fairness maximization among offline agents in OMMs, on the prototypical online matching model with \textit{known}, \textit{independent}, and \textit{identical} (KIID) arrivals, which is widely used to model the dynamic arrivals of online agents in several real-world OMMs including rideshare and crowdsourcing markets~\cite{xu-aaai-19,dickerson2018assigning,fata2019multi}.  We choose this arrival model instead of the adversarial one, because under the latter it is impossible for an online algorithm to perform well.\footnote{This can be seen as follows: Suppose there are a large number $T$ of tasks.  The first arriving worker can perform any of them.  Workers $2,\ldots,T$ can each perform a different specialized task such that there is one task that can be served only by the first arriving worker.  The online algorithm has a $1/T$ chance of correctly allocating the first worker to this task.} 

The online matching model under KIID is as follows. We have a bipartite graph $(I,J,E)$, where $I$ and $J$ represent the types of offline and online agents, respectively, and an edge $e=(i,j)$ indicates the compatibility between the offline agent (of type)  $i$ and the online agent (of type)  $j$. All offline agents are static, while online agents arrive dynamically in a certain random way.  Specially, we have an online phase consisting of $T$ rounds and during each round $t \in [T] \doteq \{1,2,\ldots, T\}$, one single online agent $\hat{j}$ will be sampled (called $\hat{j}$ arrives) with replacement such that $\Pr[\hat{j}=j]=r_j/T$ for all $j \in J$ and  $\sum_{j \in J} r_{j}/T=1$. We assume that the arrival distribution $\{r_j/T\}$ is \emph{known}, \emph{independent}, and \emph{identical} (KIID) throughout the online phase. Upon the arrival of an online agent $j$, an immediate and irrevocable decision is required: either reject $j$ or match $j$ with one offline neighbor $i$ with $i \sim j$, \ie $(i,j) \in E$.  Throughout this paper, we assume \wl that each offline agent has a unit matching capacity.\footnote{We can always make it by creating multiple copies of an offline agent if has multiple capacities.} Suppose we have a collection of groups $\cG=\{G\}$, where each group $G$ is a set of types of offline agents (possibly overlapping)  that share some demographic characteristic such as gender, race, or religion. Consider a generic online algorithm \ALG, and let $Z_i=1$ indicate that offline agent $i$ is matched (or served). We define the following two objectives: 

\noindent\textbf{Individual Fairness Maximization} (IFM): $ \min_{i \in I} \E[Z_i]$;

\noindent\textbf{Group Fairness Maximization} (GFM): $  \min_{G \in \cG} \frac{1}{|G|} \sum_{i \in G}\E[Z_i]$ with $|G|$ being the cardinality of $G$.

Here $\ifm$ denotes the minimum expected matching rate over all individual offline agents, while \gfm denotes that over all pre-specified groups of offline agents. Our goal is to design an online-matching policy such that the above two objectives are maximized. Observe that \ifm can be cast as a special case of \gfm when each group consists of one single offline type.

\xhdr{A related model: vertex-weighted online matching under KIID} (\vw).  \vw under KIID~\cite{bib:Jaillet,brubach2020online} shares almost the same setting as our model except in the objective, where each offline agent $i$ is associated with a non-negative weight $w_i$ and the objective is to maximize the expected total weight of all matched offline agents, \ie $\max \sum_{i \in I} w_i \cdot \E[Z_i]$.

One could also consider the \textit{oblivious} variant of \vw where the weights $\{w_j\}$ are a priori unknown and chosen by an adversary.
In this variant, an oblivious algorithm cannot do better than maximizing the minimum ex-ante probability $x_i$ of matching any offline agent $i$, since the adversary can always set that weight to $1$ and the other weights to $0$.
This variant is equivalent to the special case of \ifm.

\noindent \textbf{Two assumptions on the arrival setting}. (a) Integral arrival rates. Observe that for an offline agent $j$, it will arrive with probability $r_j/T$ during each of the online $T$ rounds. Thus,  $r_j$ is equal to the expected number of arrivals of $j$ during the online phase and it is called \emph{arrival rate} of $j$. In this paper, we consider integral arrival rates for all offline agents, and by following a standard technique of creating $r_j$ copies of $j$,  we can further assume \wl that all $r_j=1$~\cite{brubach2020online}. (b) $T\gg1$.
\emph{This is a standard assumption in the literature of online bipartite matching under KIID~\cite{bib:Jaillet,bib:Manshadi,bib:Haeupler,bib:Feldman}}, where the objective is typically to maximize a linear function representing the total weight of all matched edges.
We emphasize that both of these assumptions are mild in that: the interesting case is $T\gg 1$ because if $T$ is fixed then the problem, with state space $|I|^T$, can be solved to optimality; and under the assumption $T\gg 1$, the arrival rates are arbitrarily close to integers.

\section{Preliminaries and Main Contributions}

\subsection{A clairvoyant optimal and competitive ratio}

Competitive ratio (CR) is a commonly-used metric to evaluate the performance of online algorithms~\cite{mehta2012online}.  Consider an online maximization problem like ours. Let $\ALG(\cI)=\E_{S \sim \cI} [\ALG(S)]$ denote the expected performance of $\ALG$ on an input $\mathcal{I}$, where the expectation is taken over both of the randomness of the arrival sequence $S$ of online agents and that of \ALG.  
Let $\OPT(\mathcal{I})=\E_{S \sim \cI}[\OPT(S)]$ denote the expected performance of \emph{a clairvoyant optimal} who has the privilege to optimize decisions \emph{after} observing the full arrival sequence $S$. We say $\ALG$ achieves an online competitive ratio of $\rho \in [0,1]$ if $\ALG(\cI) \ge \rho \OPT(\cI)$ for all possible inputs $\cI$. Generally, the competitive ratio captures the gap in the performance between an online algorithm subject to the real-time decision-making requirement and a clairvoyant optimal who is exempt from that.  It is a common technique to use a linear program (LP) to upper bound the clairvoyant optimal (called benchmark LP), and hence by comparing against the optimal value of the benchmark LP, we can get a valid lower bound on the target competitive ratio. 

\subsection{Benchmark Linear Program}
In this paper, we use the benchmark LP as below. For each edge $(i,j) \in E$, let $x_{ij}$ be the expected number of times that the edge $(i,j)$ is matched by the clairvoyant optimal. For each $i$, let $\cN_i=\{j: (ij)\in E\}$ denote the set of offline  neighbors of $i$. Similarly, $\cN_j$ denotes the set of online neighbors of $j$. For notation convenience, we use $i \sim j$ and $j \sim i$ to denote the relation that $i$ is incident to $j$ (\ie $i \in \cN_j$) and $j$ is incident to $i$ (\ie $j \in \cN_i$), respectively.

\begin{alignat}{2}
\ifm:~~~\max & ~~\min_{i \in I} \sum_{j \sim i}x_{ij}  &&  \label{obj-1} \\
\gfm:~~~\max & ~~\min_{G \in \cG} \frac{1}{|G|} \sum_{i \in G} \sum_{j \sim i} x_{ij}  &&  \label{obj-2} \\
\vw:~~~\max & ~~ \sum_{i \in I} w_i \cdot \sum_{j \sim i} x_{ij}  &&  \label{obj-3} \\
 & \sum_{i \sim j} x_{ij} \le 1, &&~~ \forall j \in J \label{cons:j} \\ 
 &  \sum_{j \sim i}  x_{ij} \le 1,  && ~~ \forall i \in I \label{cons:i} \\
 & 0 \le \sum_{j \in S} x_{ij} \le 1-\sfe^{-|S|}, &&~~\forall S \subseteq \cN_i, |S|=O(1)\footnotemark, \forall i \in I
  \label{cons:e}
\end{alignat}
\footnotetext{Here we assume the size of $|S|$ is upper bounded by a given constant, say $K=100$, which is independent of  $T \gg 1$. }

Throughout this paper, we use \LP~\eqref{obj-1} to denote the LP with Objective~\eqref{obj-1} and Constraints~\eqref{cons:j} to~\eqref{cons:e}. Similarly for \LP~\eqref{obj-2} and \LP~\eqref{obj-3}.  Note that though Objective~\eqref{obj-1} is non-linear, we can reduce to a linear one by replacing it with $\max \lam$ and adding one constraint $\sum_{ j\sim i} x_{ij} \ge \lam$ for all $i \in I$. Similarly for Objective~\eqref{obj-2}, we can replace it with $\max \lam$ and add one extra constraint $\sum_{i \in G}\sum_{ j\sim i} x_{ij} \ge \lam \cdot |G|$ for all $G \in \cG$, where $|G|$ denotes the cardinality of group $G$. Our LPs are mainly inspired by the work~\cite{brubach2020online}. Note that the critical constraint is \eqref{cons:e}, in which we have omitted terms of size $O(1/T)$, which can affect the competitive ratio by at most $O(1/T)$~\cite{brubach2020online}. Note that all the the above three LPs can be solved polynomially even after removing the restriction on the size of $|S|$ in Constraint~\eqref{cons:e}.\footnote{Though the LPs may have exponential number of constraints after removing the restriction of a constant size of $|
S|$ independent of $T$, they are all can be solved polynomially since they admit a polynomial-time separation oracle~\cite{huang2021online}. For presentation convenience, we add that restriction and it is suffice for our analysis.}



\begin{lemma}~\label{lem:benchmark}
$\LP~\eqref{obj-1}$, $\LP~\eqref{obj-2}$, and \LP~\eqref{obj-3} are valid benchmarks for \ifm, \gfm, and \vw, respectively.
\end{lemma}

\begin{proof}
Observe that the three objectives~\eqref{obj-1},~\eqref{obj-2}, and~\eqref{obj-3} capture metrics of \ifm, \gfm, and \vw, respectively. Thus, it will suffice to justify the feasibility of all constraints for a clairvoyant optimal. Constraint~\eqref{cons:j} is valid since the total number of matches relevant to an online agent $j$ should be no larger than that of expected arrivals, which is $r_j=1$. Constraint~\eqref{cons:i} is due to that every offline agent $i$ has a unit matching capacity. Constraint\eqref{cons:e} can be justified as follows.  Consider a given $i \in I$ and a given $S \subseteq \cN_i$. Observe that $\sum_{j \in S} X_{ij}=1$ denotes the random event that $i$ is matched with one of its neighbors in $S$, whose probability should be no larger than that of at least one neighbor in $S$ arrives at least once during the online $T$ rounds. Thus,
\begin{align*}
\sum_{j \in S} x_{ij} &=\E\big[\sum_{j \in S} X_{ij}\big]=\Pr\big[\sum_{j \in S} X_{ij}=1\big] \le 1-\Pr[\mbox{none of neighbors in $S$ arrives during the $T$ rounds}]\\
&= 1-\Big(1-\frac{|S|}{T}\Big)^T=1-\sfe^{-|S|}+O(1/T). 
\end{align*}
The last equality is due to our assumption that $|S|$ is upper bounded by a given constant independent of $T \gg 1$.
\end{proof}

Before describing our technical novelties, we first explain the standard LP-based sampling approach, which has been commonly used in algorithm design for various online-matching models under known distributions, either KIID~\cite{dickerson2018assigning} or known adversarial distributions~\cite{alaei2012online} (where arrival distributions are still independent but not necessarily identical). A typical framework is as follows. We first propose a benchmark LP to upper bound the performance of a clairvoyant optimal and then solve the LP to get statistics   regarding how the clairvoyant optimal matches an online agent with its offline neighbors. After that, we use these statistics to guide online actions and transfer them to a plausible online-matching policy. For example, suppose by solving the LP, we know that the clairvoyant optimal will match each pair of offline-online agents $(i,j)$ with probability $x_{ij}$. Observe that by Constraint~\eqref{cons:j} in the benchmark LP, we have $\sum_{i \in \cN_j} x_{ij} \le r_j=1$. Thus, we can then transfer it to a simple non-adaptive matching policy as follows:  upon the arrival of online agent $j$, sample a neighbor $i \in \cN_j$ with probability $x_{ij}$ and match it if $i$ is available. 

\subsection{Overview of our techniques: LP-based sampling, boosting, and attenuation}
\label{sec:over}

A straight-forward analysis yields that the aforementioned non-adaptive sampling policy achieves $1-1/\sfe$ of the LP optimum, see \eg~\cite{haeupler2011online,brubach2020online}.
However, such a policy also achieves no better than $1-1/\sfe$, even when the LP has been tightened by Constraint~\eqref{cons:e}.\footnote{
To see this, consider a complete bipartite graph with $T$ nodes on each side.
Setting $x_{ij}=1/T$ for every edge $(i,j)$ is feasible in the tightened LP that leaves every offline node fractionally matched. However, the corresponding sampling policy would only leave each offline node matched with probability $1-1/\sfe$.
In fact, it can be seen on this example that any non-adaptive sampling policy must leave some offline node matched with probability at most $1-1/\sfe$.} To go beyond the competitive ratio of $1-1/\sfe$, we consider the following simple and natural idea.

\xhdr{Boosting}. Suppose we are at (the beginning of) time $t$ and an online agent $j$ arrives. Assume that by solving the benchmark LP, we learn a sampling distribution $\x_j=\{x_{ij} | i \in \cN_j\}$ for online agent $j$ with $\sum_{i \in \cN_j} x_{ij} \le 1$  from  the clairvoyant optimal. Let $\cN_{j,t} \subseteq \cN_j$ be the set of \emph{available} neighbors of $j$ at time $t$. Instead of non-adaptively sticking on the distribution $\x_j$ over $\cN_j$ throughout the online phase~\cite{xu-aaai-19,dickerson2019balancing}, we try to sample a neighbor $j$ from $\cN_{j,t}$ only following a boosted version of distribution $\x'_j=\{x_{ij}/\sum_{i \in \cN_{j,t}} x_{ij}\}$. In this way, we promote the chance of each available neighbor of $j$ at $t$ getting matched with $j$. Also, we can guarantee that the offline neighbor we have sampled is available at the time and dismiss the case that we sample an unavailable neighbor and have to reject $j$ ultimately, which is also a desirable feature to have in practice. This is the key idea in the algorithm we will present for \ifm (see Theorem~\ref{thm:main-1}).

Note that our boosting idea has already been proposed and tested in several practical crowdsourcing applications~\cite{aamas-19,dickerson2018assigning}. Though it proved to be helpful in some scenarios, the boosted version of LP-based sampling  is more challenging to analyze since one has to consider the adaptive behavior of the algorithm.
We are the first to show that it achieves an online competitive ratio of $0.725$, exceeding $1-1/\sfe$, for \ifm, which matches the second-best ratio as shown in~\cite{bib:Jaillet} for \vw. One can contrast this to~\cite{bib:Manshadi} who studied unweighted online matching under KIID and also introduced a boosting strategy, which is more complex. The main idea there is to generate two \emph{negatively} correlated sampling distributions from the original one, and sample two candidate neighbors for an online arriving agent.

A critical element in our analysis of the simple boosting algorithm for \ifm, however, is that an optimal LP solution places identical total mass $x_i:=\sum_{j\sim i}x_{ij}$ on each offline agent $i$. By contrast, for \gfm, an optimal LP solution may place higher mass on specific offline agents, \eg those belonging in many groups.
We show that when there is heterogeneity in the mass placed across offline agents, the boosting algorithm matches high-mass offline agents $i$ with probability no better than $(1-1/\sfe) \cdot x_i$ (see Lemma~\ref{lem:sampBub}), \ie its competitive ratio achieved is no longer better than $1-1/\sfe$.

\xhdr{Uniform vs.\ Non-uniform Boosting}.
To break this barrier of $1-1/\sfe$ for \gfm, our idea is to use boosting in a \textit{non-uniform} fashion. In the simple boosting algorithm above, for an incoming arrival, the probability of sampling unavailable neighbors was re-distributed evenly across its available neighbors.
However, our algorithm for \gfm is more likely to suffer from offline agents $i$ with high total mass $x_i$, since as explained above those agents are more difficult to match relative to its total mass $x_i$.

To accomplish this, we add an \textit{attenuation factor} to offline agents with small total mass in the LP solution, in an attempt to \textit{balance} the probability of matching every offline agent $i$ relative to its total mass $x_i$. Attenuation techniques have been used previously in several stochastic optimization problems, see, \eg stochastic knapsack \cite{ma2014improvements}, stochastic matching \cite{adamczyk2015improved,brubach2017}, and matching policy design in rideshare~\cite{feng2019linear,xuAAAI18}. Most similar to ours are the attenuation techniques used in~\cite{bib:Jaillet,brubach2020online} to overcome the barrier of $1-1/\sfe$ for \vw under KIID. However, our attenuation technique is different and novel in the following sense. Consider an online agent $j$ that arrives at time $t$ and let $\x_j=\{x_{ij}| i \in \cN_j\}$ be the sampling distribution of $j$ before attenuation. The idea in~\cite{bib:Jaillet,brubach2020online} is to carefully design factors $\alp_{ij}$ that are added directly to the sampling distribution $\x_j$ such that we finally sample an offline neighbor $i \in \cN_j$ with an updated probability $\alp_{ij} \cdot x_{ij}$ upon the arrival of $j$. In this way, they can both promote the performance of an offline agent $i$ with large mass by setting  $\alp_{ij}>1$ and compress that of $i$ with small mass by setting $\alp_{ij}<1$. In contrast, our idea is to adaptively and randomly update the set of neighbors of $j$ subject to sampling. Recall that the idea of boosting is to sample a neighbor of $j$ only from the set  of available neighbors at that time following a boosted version of distribution  $\x'_j=\{x_{ij}/\sum_{i \in \cN_{j,t}} x_{ij}\}$. Our attenuation idea is to further enhance the power of boosting by randomly ``muting'' some available neighbors of $j$ at time $t$ with small mass (\ie forcefully labeling them as unavailable) and then apply the boosting idea to the set of all available neighbors that survives the muting procedure. We expect our attenuation can help compress the performance of offline agents with small mass and promote that of offline agents with large mass as a result.

\subsection{Main contributions}

In this paper, we propose two generic online-matching based models to study individual and  group fairness maximization among offline agents in OMMs. For \ifm and \gfm, we present 
an LP-based sampling with boosting (\sbo) and another sampling algorithm with boosting and attenuation (\sab), respectively. Here are our main theoretical results.

\begin{theorem}\label{thm:main-0}
For \ifm and \gfm, both \gre and \rank achieve an online competitive ratio of 0.  
\end{theorem}

\begin{theorem}\label{thm:main-1}
A simple \LP-based sampling algorithm with boosting (\sbo) achieves an online competitive ratio at least $0.725$ for \ifm.
\end{theorem}

\begin{theorem}\label{thm:main-2}
An LP-based sampling algorithm with attenuation and boosting (\sab) achieves an online competitive ratio at least $0.719$ for \gfm and \vw.
\end{theorem}

We emphasize that Theorem~\ref{thm:main-0} requires a non-trivial analysis and makes a significant statement. Indeed, it is a priori unclear why \gre (with randomized tiebreaking) or \rank should be so poor for fairness maximization, when randomization is built into both of them. This contrasts with facts that \gre achieves a ratio \emph{equal to} $1-1/\sfe$ for vertex-weighted online matching under KIID~\cite{goel2008online} and \rank achieves $1-1/\sfe$ for unweighted online matching even under adversarial~\cite{kvv}. This shows that \ifm and \gfm are new, distinct online matching problems in which the baseline algorithms of \gre and \rank do not work.

On the other hand, the technical contributions in our LP-based algorithms/analysis have already been discussed in Section~\ref{sec:over}.
Complementing the guarantees, there is an upper bound of 0.865 due to \cite{bib:Manshadi} for unweighted KIID with integral arrival rates, which can be modified to hold for \ifm and \gfm.
We acknowledge that for \vw under KIID with integral arrival rates, the guarantee of 0.719 implied by our \sab algorithm does not improve the two state-of-the-art algorithms, which achieve ratios of $0.725$~\cite{bib:Jaillet} and $0.729$~\cite{brubach2020online}, respectively.

Nonetheless, we still compare our algorithms against \gre, \rank, and these state-of-the-art algorithms in the literature~\cite{brubach2020online,bib:Jaillet,bib:Manshadi} in real-data simulations.
Our datasets include a public ride-hailing dataset collected from the city of Chicago, using which we construct instances for \ifm and \gfm, as well as four datasets from the Network Data Repository~\cite{DBLP:conf/aaai/RossiA15}, using which we test the classical problem \vw as well. For \ifm and \gfm, simulation results show that our \sbo algorithm, as well as the algorithm from \cite{bib:Manshadi} (after being adapted to our problem), significantly outperform others. For \vw, simulation results show that \sbo, along with \gre, significantly outperform others.

This demonstrates that among algorithms which appear to perform well in practice (boosting, \gre, \cite{bib:Manshadi}), our simple boosting algorithm achieves the best guarantee, and importantly, simultaneously performs well both for fairness maximization and for weight maximization (whereas \cite{bib:Manshadi} only performs well for \ifm/\gfm while \gre only performs well for \vw).
Our simple boosting idea is also much simpler to implement than the cleverly correlated sampling in \cite{bib:Manshadi}.

As a final conceptual result, we show the fairness maximization problems to be no harder than \vw in terms of the optimal competitive ratio. For a given model, let $\Phi(\cdot)$ denote the online competitive ratio achieved by an optimal online algorithm, \ie the best competitive ratio possible.  We establish the following, where the first inequality is trivial since \ifm is a special case of \gfm.

 \begin{theorem}\label{thm:main-4}
$\Phi(\ifm)\ge\Phi(\gfm)\ge\Phi(\vw)$.
\end{theorem}

Of course, this result is on the theoretically optimal online algorithms, which cannot be computed.  The fairness maximization problems are computationally distinctive from \vw.

\xhdr{Roadmap}. In Section~\ref{sec:thm1}, we present the algorithm \sbo for \ifm and prove Theorem~\ref{thm:main-1}; in Section~\ref{sec:thm2}, we present the algorithm \sab for \ifm and \vw and prove Theorem~\ref{thm:main-2}; in Section~\ref{sec:thm1-4}, we prove Theorems~\ref{thm:main-0} and~\ref{thm:main-4}, and in Section~\ref{sec:exp}, we present details regarding our real datasets and relevant experimental results.

\section{Other Related Works} In recent years, online-matching-based models have seen wide applications ranging from blood donation~\cite{McElfresh2020MatchingAF} to volunteer crowdsourcing~\cite{Manshadi2020OnlinePF} and from kidney exchange~\cite{Li2019IncorporatingCP} to rideshare~\cite{xuAAAI18}. Here we briefly discuss a few studies that investigate the fairness issue. Both works of~\cite{suhr2019} and~\cite{lesmana2019} have studied the income inequality among rideshare drivers. However, they mainly considered a complete offline setting where  the information of all agents in the system including drivers and riders is known in advance. They justified that by focusing a short window and thus, all agents can be assumed offline. There are several other works that considered fairness in matching in an offline setting where all agents' information is given as part of the input, see, \eg \cite{sankar2021matchings,garcia2020fair}.~\citet{nanda2020balancing} proposed a bi-objective online-matching-based model to study the tradeoff between the system efficiency (profit) and the fairness among rideshare riders during high-demand hours. In contrast, ~\citet{xu2020trade} presented a similar model to examine the tradeoff between the system efficiency and the income equality among rideshare drivers. Unlike focusing on one single objective of fairness maximization like here, both studies in~\cite{nanda2020balancing} and~\cite{xu2020trade} seek to  balance the objective of fairness maximization with that of profit maximization. Recently,~\citet{ma2021grouplevel} considered a similar problem to us but focus on the fairness among online agents.~\citet{manshadi2021fair} studied fair online rationing such that each arriving agent can receive a fair share of resources proportional to its demand. The fairness issue has been studied in other domains/applications as well, see, \eg online selection of candidates~\cite{salem2019closing}, influence maximization~\cite{tsang2019group}, bandit-based online learning~\cite{patil2020achieving,gillen2018online,joseph2016fairness}, online resource allocation~\cite{sinclair2021sequential,bansal2021online}, and classification~\cite{dwork2012fairness}.

\input{3-fm}

\input{4-vw}

\input{5-thm-2}

\input{6-hard}

\input{6-exp}

\section{Conclusions and Future work}\label{sec:con}
In this paper, we proposed two online-matching based models to study individual and  group fairness maximization among offline agents in OMMs. For individual and group fairness maximization, we presented two LP-based sampling algorithms, namely \sbo and \sab, which achieve online competitive ratios at least $0.725$ and $0.719$, respectively.  We conducted extensive numerical experiments and results show that \sbo is not only conceptually easy to implement but also highly effective in practical instances of fairness-maximization related models. One interesting future direction is to show some explicit upper bounds for \ifm, \gfm, and \vw. So far, all existing upper bounds for online matching under KIID 
are for the unweighted case due to~\cite{bib:Manshadi}. Can we derive some upper bounds specifically for \ifm, \gfm or \vw? We expect the upper bound of \ifm should be slightly higher than that of \vw as suggested by Theorem~\ref{thm:main-4}.

\newpage
{
\bibliographystyle{unsrtnat}

\bibliography{EC_21}
}
\newpage
\appendix

\end{document}

%% file: 3-fm.tex
\section{Individual Fairness Maximization (\ifm)}\label{sec:thm1}
 Let $\cN_{j,t}$ denote the set of available neighbors of $j$ at $t$. For the ease of notation, we use $\{x_{ij}\}$ to denote an optimal solution to \LP~\eqref{obj-1} when the context is clear. Let $ x_i\doteq \sum_{j \sim i} x_{ij}$  for each $i \in I$. Assume w.l.o.g.  that $x_i=\tau$ for all $i \in I$, where $\tau \in [0,1]$  is the the optimal value of \LP~\eqref{obj-1}.\footnote{We can always make it by decreasing all $\{x_{ij}| j \sim i\}$ for those $i$ with $x_i>\tau$ without affecting the optimal LP value.} Our LP-based sampling with boosting is formally stated as follows.

\begin{algorithm}[h!]
\DontPrintSemicolon
\textbf{Offline Phase}: \\
Solve  \LP~\eqref{obj-1} and let $\{x_{ij}\}$ be an optimal solution.\;
\textbf{Online Phase}:\\
 \For{$t=1,\ldots,T$}{
 Let an online agent of type $j$ arrive at (the beginning of) time $t$. \;
Let $\cN_{j,t}=\{i \in \cN_j, i \mbox{ is available at $t$}\}$, \ie the set of available neighbors of $j$ at $t$. \;
\eIf{$\cN_{j,t}=\emptyset$}{ Reject $j$.}
{Sample a neighbor $i \in \cN_{j,t}$ with probability $x_{ij}/\sum_{i' \in \cN_{j,t}}x_{i',j}$. \label{alg:step3}
}
}
\caption{Sampling with Boosting (\sbo).}
\label{alg:adap}
\end{algorithm}

\vspace{-0.5in}
\subsection{Proof of the main Theorem~\ref{thm:main-1}}
For an offline agent $i \in I$, let $Z_i=1$ indicate that $i$ is matched in the end in \sbo. The key idea is to show the below theorem. 

\begin{theorem}\label{thm:1-a}
$\E[Z_i] \ge 0.725 \cdot \tau$ for all $i \in I$. 
\end{theorem}
Theorem~\ref{thm:1-a} suggests that each offline agent gets matched in \sbo with probability at least $0.725 \cdot \tau$. Thus, we have $\min_i \E[Z_i] \ge  0.725 \cdot \tau \ge 0.725 \cdot \OPT$ by Lemma~\ref{lem:benchmark}, where $\OPT$ denotes the performance of a clairvoyant optimal. Therefore, we establish the main Theorem~\ref{thm:main-1}.

For each offline agent $i \in I$ and $t \in [T]$, let $\bo_{i,t}=1$ indicate that $i$ is available at (the beginning of) $t$ in \sbo, and $q_{i,t}$ be the probability that $i$ is matched during $t$ conditioning on $i$ is available at (the beginning of) $t$, \ie $q_{i,t}=\Pr[\bo_{i,t+1}=0| \bo_{i,t} ]$. Recall that $Z_i=1$ indicate that $i$ is matched in the end in \sbo. Thus, 
\begin{equation}\label{eqn:1-a}
\E[Z_i]=1-\prod_{t=1}^T (1-q_{i,t}).
\end{equation}

Recall that in the optimal solution, we have $x_i=\tau$ for all $i \in I$. 
\begin{lemma}\label{lem:1-a}
For any $i$ and $i'$ with $i \neq i'$ and any time $t$, we have 
$\Pr[ \bo_{i',t}=1| \bo_{i,t}=1] \le (1-\tau/T)^{t-1}$.
\end{lemma}
\begin{proof}
The online sampling process in \sbo can be interpreted through the following balls-and-bins model:  Each offline agent $i$ corresponds to a bin $i$ and each edge $e=(ij)$ corresponds to a ball $e$; During each round $t \in [T]$, a ball $(ij)$ will arrive with probability $0$ if $i$ gets occupied before $t$, and a ball $(ij)$ will arrive with probability $(x_{ij}/X_{j,t})\cdot(1/T)$ and shoot the bin $i$ if bin $i$ is empty (not occupied) at $t$, where $X_{j,t}\doteq \sum_{i' \in \cN_{j,t} } x_{i',j}$. Here $\cN_{j,t}$ can be viewed as the set of unoccupied bin $i$ at $t$ with $i \sim j$. Observe that for any $j$ and $t$, we have
$X_{j,t}=\sum_{i' \in \cN_{j,t} } x_{i',j} \le \sum_{i' \sim j } x_{i',j} \le 1$ due to Constraint~\ref{cons:j} in \LP~\eqref{obj-1}.

Observe that $\bo_{i,t}=1$ suggests that for any round $t'<t$, none of the balls $e=(ij)$ with $j \sim i$ arrives during $t'$. Consider a given $t'<t$ and a given $i' \neq i$. Assume $\bo_{i,t}=1$ and $i'$ is not occupied at $t'$. Then, we see that each ball $e=(i',j')$ with $j' \sim i'$ will arrive and shoot $i'$ with probability at least $(x_{i',j'}/X_{j',t'})\cdot(1/T)>x_{i',j'}/T$ since $X_{j',t'} \le 1$ for all $j' \sim i'$ and $t'<t$. This implies that the probability that none relevant balls will shoot $i'$ during $t'$ should be at most $1-\sum_{j' \sim i'} x_{i',j'}/T=1-\tau/T$. Here we invoke our assumption that every offline agent $i'$ has $x_{i'}=\sum_{j' \sim i'} x_{i',j'}=\tau$  in the optimal solution. Therefore, we claim that $i'$ will remain unoccupied after $t-1$ rounds with probability at most $(1-\tau/T)^{t-1}$. 
\end{proof}

\begin{proof}[Proof of Theorem~\ref{thm:1-a}]
Focus on a given offline agent $i^*$. For the ease of notation, we drop the subscription of $i^*$, and use $q_t,\bo_t$, and $Z$ to denote the corresponding values with respect to $i^*$. 


Now, we try to lower bound the value of $q_t$. Consider a given $t$. For each $i \neq i^*$, recall that $\bo_{i,t}=1$ indicate that $i$ is available at $t$. By the nature of \sbo, we see that conditioning on $i^*$ is available at $t$ ($\bo_t=1$), $i^*$ will be matched during $t$ iff one of its neighbors $j \sim i^*$ arrives and $(i^*,j)$ gets sampled.   Recall that for each $j \sim i^*$, $\cN_{j,t}$ denotes the set of available neighbors incident to $j$ at $t$, and $X_{j,t}=\sum_{i \in \cN_{j,t}}x_{ij}$. Observe that
\begin{equation}\label{eqn:1-c}
q_t =\E\bB{\frac{1}{T}\sum_{j \sim i^*} \frac{x_{i^*,j}}{X_{j,t}}\bigg|  \bo_t=1} \ge  \frac{1}{T}\sum_{j \sim i^*}   \frac{x_{i^*,j}}{ \E[X_{j,t}| \bo_t=1]}.
\end{equation}
The last inequality above is due to Jensen's inequality and the convexity of function $1/x$. Note that

\begin{align*}
 \E[X_{j,t}| \bo_t=1] &=\E\bb{\sum_{i \in \cN_{j,t}}x_{ij}| \bo_t=1}
 =x_{i^*,j}+\sum_{i \neq i^*, i\sim j} x_{i,j} \cdot \E[\bo_{i,t}| \bo_t=1] \\
 & \le x_{i^*,j}+ \sum_{i \neq i^*, i\sim j} x_{i,j} \cdot (1-\tau/T)^{t-1}.~~\mbox{(By Lemma~\ref{lem:1-a})}\\
 &\le  x_{i^*,j}+(1-x_{i^*,j}) \cdot (1-\tau/T)^{t-1}.~~\mbox{(Due to Constraint~\eqref{cons:j} of \LP~\eqref{obj-1})}
\end{align*}
Substituting the above inequality to Inequality~\eqref{eqn:1-c}, we have

\begin{align*}
 q_t& \ge \frac{1}{T}\sum_{j \sim i^*}   \frac{x_{i^*,j}}{x_{i^*,j}+ (1-x_{i^*,j}) \cdot (1-\tau/T)^{t-1}}.
 \end{align*}
Plugging the above results into Equation~\eqref{eqn:1-a}, we have
\begingroup
\allowdisplaybreaks
\begin{align*}
\E[Z]&=1-\prod_{t=1}^T (1-q_t)\\
& \ge 1-\prod_{t=1}^T \bP{1-\frac{1}{T}\sum_{j \sim i^*}   \frac{x_{i^*,j}}{x_{i^*,j}+ (1-x_{i^*,j}) \cdot (1-\tau/T)^{t-1}}}\\
&=1-\exp\bP{\sum_{t=1}^T \ln\bP{1-\frac{1}{T}\sum_{j \sim i^*}   \frac{x_{i^*,j}}{x_{i^*,j}+ (1-x_{i^*,j}) \cdot (1-\tau/T)^{t-1}}} }\\
&\ge 1-\exp\bP{-\sum_{t=1}^T\frac{1}{T}\sum_{j \sim i^*}   \frac{x_{i^*,j}}{x_{i^*,j}+ (1-x_{i^*,j}) \cdot (1-\tau/T)^{t-1}}  }\\
&=1-\exp\bP{-\sum_{j \sim i^*} \sum_{t=1}^T\frac{1}{T}  \frac{x_{i^*,j}}{x_{i^*,j}+ (1-x_{i^*,j}) \cdot (1-\tau/T)^{t-1}}  }\\
&=1-\exp\bP{-\sum_{j \sim i^*} \int_0^1 d\zeta \frac{x_{i^*,j}}{x_{i^*,j}+ (1-x_{i^*,j}) \cdot\sfe^{-\tau \cdot \zeta}}  } ~~\mbox{(Taking $T\rightarrow \infty$)}\\ 
&=1-\exp\bp{-\frac{1}{\tau}\sum_{j \sim i^*} \ln \bp{1+x_{i^*,j} \cdot(\sfe^\tau-1)}}
\end{align*}
\endgroup
The second inequality above is due to the fact that $\ln(1-x) \le -x$ for all $x \in[0,1)$.

Let $g_\tau(x)=\ln (1+x(\sfe^\tau-1))$, where $\tau \in [0,1]$ is a parameter. To get a lower bound for $\E[Z]$, we need to solve the below minimization program. For the ease of notation, we omit the subscription of $i^*$ and use $x_j \doteq x_{i^*,j}$. 
\begin{equation}\label{eqn:1-b}
\left\{\min \sum_{j \sim i^*} g_\tau(x_j):~~\sum_{j \sim i^*}x_j =\tau; \sum_{j \in S} x_j \le 1-\sfe^{-|S|}, \forall S \subseteq \cN_{i^*}, |S|=O(1).  \right\}
\end{equation}

Note that the first constraint is due to our assumption $x_{i^*}=\tau$; the rest is due to Constraint~\eqref{cons:e} in the benchmark LP. By Lemma~\ref{lem:1-b}, we have $\sum_{j \sim i^*} g(x_j) \ge \sum_{x \in \cS(\tau)} g(x)$. 
Plugging this result to the inequality on $\E[Z]$ above, we get 
\[
\E[Z] \ge 1-\exp\bp{-\frac{1}{\tau}\sum_{x \in \cS(\tau)} \ln \bp{1+x \cdot(\sfe^\tau-1)}},
\]
where $\cS(\tau)$ is an optimal solution to Program~\eqref{eqn:1-b} as defined in Lemma~\ref{lem:1-b}. We can verify that $\E[Z]/\tau$ gets minimized at $\tau=1$ and the corresponding ratio is $0.725$. Thus, we are done.
\end{proof}

An optimal solution to the minimization program~\eqref{eqn:1-b} can be characterized by the following set $\cS(\tau)$. Let $\ell_\tau$ be the largest integer satisfying that $1-\sfe^{-\ell_\tau}\le \tau$ with $\ell_0=0$ and $\ell_1=\infty$.

\begin{lemma}\label{lem:1-b}
An optimal solution to the minimization program~\eqref{eqn:1-b} can be characterized by $\cS(\tau)=\{ 1-\sfe^{-1}, \sfe^{-1}-\sfe^{-2}, \cdots, \sfe^{-\ell_\tau+1}-\sfe^{-\ell_\tau}, \tau-(1-\sfe^{-\ell_{\tau}}) \}$.
\end{lemma}
\begin{proof}
 Recall that $g_\tau(x)=\ln (1+x(\sfe^\tau-1))$. We can verify that $g_\tau(0)=0$ and $g_\tau(1)=\tau$, and it is an increasing concave function over $x\in (0,1)$ for all $\tau\in [0,1]$. The main idea in our proof is a local perturbation analysis. In the following, we focus on the special case when $\tau=1$ and $|S| \le 2$ in Constraint~\eqref{cons:e}. All the rest follows a similar analysis. 
 
 Assume $\tau=1$ and let $g(x)=\ln (1+x(\sfe-1))$. Suppose have an optimal solution $\{x_j\}$ to Program~\eqref{eqn:1-b}. WLOG assume that $x_1 \ge x_2 \ge \cdots$. By Constraint~\eqref{cons:e} with $|S|=1$, we have $x_1 \le 1-\sfe^{-1}$. Suppose $x_1<1-\sfe^{-1}$. Then, we can apply a local perturbation as $x'_1=x_1+\ep$ and $x'_2=x_2-\ep$ with $\ep=\min(1-\sfe^{-1}-x_1, x_2, \frac{1}{4}-\frac{3}{4\sfe^2})$, and $x'_\ell=x_\ell$ for all $\ell>2$. Now, we try to show that (a)   $g(x'_1)+g(x'_2)<g(x_1)+g(x_2)$ and (b) $\x'=\{x'_j\}$ is still feasible to Program~\eqref{eqn:1-b}.

For Point (a), we have
\begin{align*}
g(x'_1)+g(x'_2)-g(x_1)-g(x_2)&=g(x_1+\ep)-g(x_1)+g(x_2-\ep)-g(x_2)\\
&=\ep \bP{\frac{g(x_1+\ep)-g(x_1)}{\ep}-\frac{g(x_2)-g(x_2-\ep)}{\ep}}<0
\end{align*}
For Point (b): we can verify that $\x'$ will satisfy the first two constraints of Program~\eqref{eqn:1-b} in the way that $\sum_j x'_j=1$ and $x'_j \le 1-\sfe^{-1}$ for all $j$. Note that for any $j \ge 3$,
\begin{align*}
x'_1+x'_j & \le x_1+\ep+x_3 \le \frac{x_1+x_3+x_1+x_2}{2}+\ep\le \frac{1}{2}+\frac{x_2}{2}+\ep \le \frac{1}{2}+\frac{1}{2} \frac{1-\sfe^{-2}}{2}+\ep \le 1-\sfe^{-2}.
\end{align*}
Thus, we prove Points (a) and (b), which yields a contradiction with our assumption that $x_1<1-\sfe^{-1}$. Repeating the above analysis, we can show that in the optimal solution, $x_2=\sfe^{-1}-\sfe^{-2}$, and $x_3=\sfe^{-2}$ for the case $\tau=1$ and $|S| \le 2$ in Constraint~\eqref{cons:e}. 
\end{proof}

%% file: 4-vw.tex
\section{Group fairness maximization and agent-weighted matching}\label{sec:thm2}

\subsection{Motivation for attenuation} 
We first give an example showing that \sbo can never beat $1-\sfe^{-1}$ for \vw without attenuation. 
\begin{example}\label{exam:1}
Consider such a bipartite graph $(I,J,E)$ as follows.  Recall that by KIID assumption with all unit arrival rates, we have $T=n=|J|$.  The set of neighbors of $j$, denoted by $\cN_j$, satisfies the property that (1) $|\cN_j|=n, \forall j \in J$; (2) $\cap_{j \in J} \cN_j =\{i^*\}$. In other words, each $j$ has a set of $n$ neighbors and they are almost disjoint except sharing one single offline agent $i^*$. Thus, under this setting, we have (1) $|I|=m=n(n-1)+1$; (2) $i^*$ has $J$ as the set of neighbors and  every $i \neq i^*$ has one single neighbor in $J$. Let $w_{i^*}=1$ and $w_i=\ep^3$ with $\ep=1/n$ for all $i \neq i^*$. We can verify that an optimal solution to $\LP~\eqref{obj-1}$ can be as follows: $x_{ij}=\ep$ for all $(ij)\in E$. Under this solution, we have $x_j =\sum_{i \sim j} x_{ij}=1$ for all $j$, and $x_{i^*}=\sum_{j \sim i^*}x_{i^*,j}=1$ and $x_i=\sum_{j \sim i}x_{ij}=\ep$ for all $i \neq i^*$. The optimal LP value is $1+\ep^4(1/\ep^2-1/\ep)=1+\ep^2-\ep^3$.
\end{example}

\begin{lemma} \label{lem:sampBub}
\sbo can never beat the ratio of $1-\sfe^{-1}+o(1)$  on Example~\ref{exam:1}, where $o(1)$ is a vanishing term when $T \rightarrow \infty$.
\end{lemma}

\begin{proof}
Consider $i^*$ and let $Z_{i^*}=1$ indicate that $i^*$ is matched in \sbo in the end. Observe that during every round $t$, one $j \sim i^*$ will be sampled uniformly with probability $1/n$ and *shoot* one available neighbor $i \in  \cN_{j,t}$. Let $N_{j}=|\cN_{j,T+1}|$ be the number of available neighbors incident to $j$ surviving in the end, and let $M_j=n-N_j$, which refers to the number of neighbors of $j$ got shot. Observe that we have $T=n$ online arrivals and every arrival will shoot one item in $\cN_j$ uniformly over all $j \in J$. This process can be interpreted as a balls-and-bins model where we have $n$ balls and $n$ bins, and thus, $M=\max_j M_j$ can be viewed as the largest bin load. From~\cite{mitzenmacher2017probability}, we see that with probability $1-1/n$, the largest bin load is $M=\Theta(\ln n/ \ln \ln n) \le \ln n$ when $n$ is sufficiently large. 

Let $\SF$ be the event that $M \le \ln n$. Assume $\SF$ occurs. We see that for all $j \sim i^*$ and $t \in [T]$, $X_{j,t}=\sum_{i \in \cN_{j,t}}x_{ij} \ge \ep (n-M)=\ep(n-\ln n)=1-\ln n/n$. Let $Z_{i^*}=1$ indicate that $i^*$ is matched in the end. We have 
\begin{align*}
\E[Z_{i^*} |\SF] &= 1-\prod_{t=1}^T \bP{1-\sum_{j \sim i^*} \frac{1}{T}\cdot\frac{x_{i^*,j}}{X_{j,t}} } \le  1-\prod_{t=1}^T \bP{1-\sum_{j \sim i^*} \frac{1}{T}\cdot\frac{\ep}{1-\ln n/n} }\\
&=1- \bP{1-\frac{1}{T}\cdot\frac{1}{1-\ln n/n} }^T
\le 1- \exp\bp{-\frac{1}{1-\ln n/n}-\frac{1}{T}\frac{1}{(1-\ln n /n)^2}}\\
&\le 1-\sfe^{-1}+O(\ln n/n).
\end{align*}
Therefore, we have that 
\[
\E[Z_{i^*}] = \E[Z_{i^*} |\SF]\cdot \Pr[\SF]+ \E[Z_{i^*} |\neg \SF]\cdot \Pr[\neg \SF] \le  1-\sfe^{-1}+O(\ln n/n).
\]
Recall that the optimal LP value is $1+\ep^4(1/\ep^2-1/\ep)=1+\ep^2-\ep^3$. For each $i \neq i^*$, let $Z_i=1$ indicate that $i$ is matched in \sbo in the end. 
Observe that  the expected total values obtained by \sbo should be at most 
\[
\sum_{i \neq i^*} w_{i} \cdot\E[Z_i]+w_{i^*} \cdot\E[Z_{i^*}] \le \ep^3 \cdot n^2+ \E[Z_{i^*}] \le  1-\sfe^{-1}+O(\ln n/n).
\]
Thus, the final competitive ratio of \sbo on Example~\ref{exam:1} should be at most 
\[ \frac{1-\sfe^{-1}+O(\ln n/n)}{1+\ep^2-\ep^3} \le 1-\sfe^{-1}+O(\ln n/n).
\]
\end{proof}

\subsection{An LP-based sampling algorithm with attenuation and boosting (\sab)}

For a given offline agent $i$, we say $i' \in I$ is an \emph{offline neighbor} of $i$ iff there exists one online agent of $j$ such that $j \sim i$ and $j \sim i'$. Let $\cS_i$ be the set of offline neighbor of $i$. Example~\ref{exam:1} suggests that when all offline neighbors of $i$ have very small values in the optimal solution, the boosting strategy shown in \sbo will have little effect on improving the overall matching probability of $i$. Observe that for each offline vertices $i\neq i^*$ on Example~\ref{exam:1}, it will be matched with a probability at least $\E[Z_{i}] \ge 1-\sfe^{-\ep} \sim \ep=x_i$. In other words, the chance of getting matched for every $i \neq i^*$ in \sbo almost matches its contribution in the LP solution. In contrast, the chance that $i^*$ is matched is only a fraction of $1-\sfe^{-1}$ of its contribution in the LP solution. These insights motivate us to add appropriate attenuations to those unsaturated offline vertices such that the boosting strategy can work properly for those saturated ones.  

\xhdr{Offline-phase simulation-based attenuation}. Let us first introduce two auxiliary states for offline vertices, called \emph{active} and \emph{inactive}, which are slightly different from available (not matched) and unavailable (matched) as shown before. In our attenuation framework, we assume all offline vertices are active at the beginning ($t=1$). When an active offline agent $i$ is matched, we will label it as inactive. Meanwhile, we need to forcefully mute some active offline agent, label it as inactive, and view it as being virtually matched. Consider the instance on Example~\ref{exam:1}: In order to make the boosting strategy work for the dominant agent $i^*$, we have to intentionally label those active non-dominant vertices $i$ as inactive such that the sampling probability of $i^*$ can be effectively promoted when some $j \in \cN_i \cap \cN_{i*}$ arrives. Note that the transition from being active to inactive is \emph{irreversible}: Once an active offline agent $i$ is labeled as inactive, it will stay on that state permanently.

Here are the details of our simulation-based attenuation. By simulating all online steps of \sab up to time $t$, we can get a very sharp estimate of the probability that each $i$ is active at $t$, say $\alp_{i,t}$. If $\alp_{i,t} \le (1-1/T)^{t-1}$, then no attenuation is needed at $t$. Otherwise, add an attenuation factor of $(1-1/T)^{t-1}/\alp_{i,t}$ to agent $i$ at $t$ as follows: If $i$ is active at $t$, then label $i$ as \emph{inactive} with probability  $1-(1-1/T)^{t-1}/\alp_{i,t}$ and keep it active with probability $(1-1/T)^{t-1}/\alp_{i,t}$. In this way, we decrease the probability of $i$ being active at $t$ to the target $(1-1/T)^{t-1}$. The above attenuation can be summarized as follows: If $i$ is available at $t$, then label $i$ as  active and inactive with respective probabilities  $\beta_{i,t}$ and $1-\beta_{i,t}$, where $\beta_{i,t}=\min\bp{1, (1-1/T)^{t-1}/\alp_{i,t}}$. 

\textbf{Here are a few notes on the simulation-based attenuation scheme}. (1) When computing the attenuation factor $\beta_{i,t}$ for $i$ at $t$, we should simulate all online steps of \sab up to $t$ that include applying all the attenuation factors as proposed during all the rounds before $t$. (2) During every round, we apply the corresponding attenuation factor to each active offline agent in an independent way. (3) All attenuation factors can be computed in an \emph{offline} manner, \ie before the online phase actually starts.


\begin{algorithm}[h!]
\DontPrintSemicolon
\textbf{Offline Phase}: \;
\tcc{The offline phase will take as input $\{(I,J,E), \{w_i\}, \{r_j\},T\}$, and output $\{\beta_{i,t}\}$, where $\beta_{i,t}$ denotes the attenuation factor applied to an offline agent $i$ during round $t$.}
Solve  \LP~\eqref{obj-1} and let $\{x_{ij}\}$ be an optimal solution.\;
\emph{Initialization}: When $t=1$, set $\beta_{i,t}=1$ for all $i \in I$.\;
\For{$t=2,3,\cdots,T$}
{Applying Monte-Carlo method to simulate Step~\ref{alg:on-1} to Step~\ref{alg:on-2} for all the rounds $t'=1,2,\cdots,t-1$ of Online Phase, we get a sharp estimate of the probability that each offline agent $i$ is active at (the beginning of) $t$, say $\alp_{i,t}$. \;
Set $\beta_{i,t}=\min\bp{1, (1-1/T)^{t-1}/\alp_{i,t}}$.}
\textbf{Online Phase}:\;
\emph{Initialization}: Label all offline vertices \emph{active} at $t=1$.\;
 \For{$t=1,\ldots,T$}{

\emph{Independently} relabel each \emph{active} offline agent $i$ as active and inactive with respective probabilities $\beta_{i,t}$ and $1-\beta_{i,t}$. \label{alg:on-1}\;
Let an online agent of type $j$ arrive at time $t$. Let $\cN_{j,t}=\{i \in \cN_j, i \mbox{ is active at $t$}\}$, \ie the set of active neighbors of $j$ at $t$.  \;
\eIf {$\cN_{j,t}=\emptyset$}{Reject $j$.}
{Sample a neighbor $i \in \cN_{j,t}$ with probability $x_{ij}/\sum_{i' \in \cN_{j,t}}x_{i',j}$ and label $i$ as inactive. \label{alg:on-2}}
}
\caption{Sampling with Attenuation and Boosting (\sab).}
\label{alg:adap}
\end{algorithm}

%% file: 5-thm-2.tex
\subsection{Proof of the main Theorem~\ref{thm:main-2}}
Similar to the proof of  Theorem~\ref{thm:main-1}, we aim to show that each offline agent $i$ will be matched in \sab with a probability $\E[Z_i] \ge 0.719 \cdot x_{i}$, where $Z_i=1$ indicates that $i$ is matched in \sab, and $x_i=\sum_{j \sim i}x_{ij}$ is the total mass allocated to $i$ in the optimal LP solution. This will suffice to prove  Theorem~\ref{thm:main-2}. The argument is as follows. (1) For \gfm, we have $\frac{1}{|G|} \sum_{i \in G} \E[Z_i] \ge \frac{0.719}{|G|} \sum_{i \in G} x_i$ for all $G \in \cG$. This suggests that $\sab=\min_{G \in \cG} \frac{1}{|G|} \sum_{i \in G} \E[Z_i] \ge 0.719 \cdot \min_{G \in \cG} \sum_{i \in G} x_i/|G| \ge 0.719 \cdot \OPT$ due to Lemma~\ref{lem:benchmark}, where $\sab$ and $\OPT$ refer to the performance of \sab and a clairvoyant optimal, respectively. (2) For \vw, we have $\sab=\sum_{i \in I} w_i \cdot \E[Z_i] \ge 0.719 \cdot \sum_{i \in I} w_i \cdot x_i \ge  0.719 \cdot \OPT$.

For each offline agent $i$, let $\bo'_{i,t}=1$ and $\bo_{i,t}=1$ indicate that $i$ is active at $t$ \emph{before} and \emph{after} the attenuation procedure shown in Step~\ref{alg:on-1} prior to the sampling process.  Let $\alp_{i,t}=\E[\bo'_{i,t}]$ and $\gam_{i,t}=\E[\bo_{i,t}]$.  Let $q_{i,t}$ be the probability that $i$ is matched during $t$ conditioning on $i$ is active at $t$ after attenuation, \ie $q_{i,t}=\Pr[\bo'_{i,t+1}=0| \bo_{i,t}=1]=1-\E[\bo'_{i,t+1}| \bo_{i,t}=1]$. According to our attenuation, for all $i$ and $t$, we have
\begin{equation}\label{eqn:thm2-a}
\gam_{i,t}=\alp_{i,t} \cdot \beta_{i,t},~~\beta_{i,t}=\min\bp{1, (1-1/T)^{t-1}/\alp_{i,t}}, ~~\alp_{i,t+1}=\gam_{i,t}\cdot (1-q_{i,t}). 
\end{equation}
Observe that $\alp_{i,1}=\beta_{i,1}=1$ for all $i$, and $\gam_{i,t} \le (1-1/T)^{t-1}$ for all $i$ and $t$.  Though our definition of $\{\bo_{i,t}\}$ is slightly different from before, Lemma~\ref{lem:1-a} of Section~\ref{sec:thm1} still works here. 



\begin{lemma}\label{lem:2-a}
For any $i$ and $i'$ with $i \neq i'$ and any time $t$, we have 
$\Pr[ \bo_{i',t}=1| \bo_{i,t}=1] \le (1-1/T)^{t-1}$.
\end{lemma}
\begin{proof}
The proof of Lemma~\ref{lem:1-a} in Section~\ref{sec:thm1} suggests that $\bo'_{i,t}$ and $\bo'_{i',t}$ are negatively correlated before attenuation. Thus, we have $\Pr[\bo'_{i',t}=1| \bo'_{i,t}=1] \le \Pr[\bo'_{i',t}=1]$. Observe that attenuation factors are applied independently to all offline vertices. Therefore,
\[
\Pr[ \bo_{i',t}=1| \bo_{i,t}=1] =\beta_{i',t} \cdot \Pr[ \bo'_{i',t}=1| \bo'_{i,t}=1] \le
\beta_{i',t} \cdot  \Pr[ \bo'_{i',t}=1]=\beta_{i',t} \cdot \alp_{i',t} \le (1-1/T)^{t-1}.
\]
\end{proof}

Consider a given offline agent $i^*$ with a fixed value of $x_{i^*}\doteq \sum_{j \sim {i^*}} x_{i^*,j}$. For the ease of notation, we drop the subscription of $i^*$ and use $\bo'_{t}, \bo_t$, $\alp_t$, $\beta_t$, $\gam_t$, and $q_t$ to denote the corresponding values relevant to $i^*$.  Here are a few properties of $\{q_t\}$.

\begin{lemma}\label{lem:thm2-a}
\textbf{(P1)}: $q_t \le q_{t+1}$, $\forall t \ge 1$; \textbf{(P2)} $q_t \ge \frac{1}{T}\sum_{j \sim i^*} \frac{x_{i^*,j}}{x_{i^*,j}+(1-x_{i^*,j})\cdot (1-1/T)^{t-1}}$, $\forall t \ge 1$.
\end{lemma}
\begin{proof}
Recall that for each $j \sim i^*$, $\cN_{j,t}$ denotes the set of \emph{active} neighbors of $j$ at $t$ right after attenuation. Let $X_{j,t}=\sum_{i \in \cN_{j,t}} x_{ij}=\sum_{i \sim j} x_{ij} \cdot \bo_{i,t}$. Observe that 
\begin{align}\label{eqn:2-3}
q_t =\E\bB{\frac{1}{T}\sum_{j \sim i^*} \frac{x_{i^*,j}}{X_{j,t}}\bigg|  \bo_t=1} =\E\bB{\frac{1}{T}\sum_{j \sim i^*} \frac{x_{i^*,j}}{\sum_{i \neq i^*, i \sim j} x_{i,j} \cdot \bo_{i,t}+ x_{i^*,j}}\bigg|  \bo_t=1}.
\end{align}
Observe that for each given $i \neq i^*$, $\{\bo_{i,t}| t=1,2,\ldots,T\}$ will be a non-increasing sequence due to the  irreversibility of the transition from active to inactive of $i$. Therefore, we claim that $\{q_t|t=1,\ldots,T \}$ is a non-decreasing series. Thus, we prove \textbf{(P1)}. From Equation~\eqref{eqn:2-3}, we have  
\[
q_t \ge \frac{1}{T} \sum_{j \sim i^*} \frac{x_{i^*,j}}{\sum_{i \neq i^*, i \sim j} x_{i,j} \cdot \E[ \bo_{i,t}| \bo_t]+ x_{i^*,j}} \ge \frac{1}{T} \sum_{j \sim i^*} \frac{x_{i^*,j}}{(1-x_{i^*,j}) \cdot (1-1/T)^{t-1}+ x_{i^*,j}}.
\]
The first inequality is due to Jensen's inequality and convexity of the function $1/x$. The second one follows from Lemma~\ref{lem:2-a} and the fact of $\sum_{i \sim j} x_{ij} \le 1$ due to Constraint~\eqref{cons:j}. We get \textbf{(P2)}.
\end{proof}

\textbf{(P1)} in Lemma~\ref{lem:thm2-a} suggests that $\{q_t\}$ is a non-decreasing sequence. Let $K \in [T]$ be such a turning point that $q_{K-1}<1/T$ and $q_{K} \ge 1/T$. 

\begin{lemma}\label{lem:2-c}
For each $1<t \le K$, we have $\beta_{t}<1$ and $\gam_t=(1-1/T)^{t-1}$, and for each $t > K$, $\beta_{t}=1$.
\end{lemma}
\begin{proof}
By \textbf{(P1)}, we have $q_1 \le q_2 \le \cdots \le  q_{K-1}<1/T$. Observe that 
$\alp_1=\beta_1=\gam_1=1$. Now  we consider $t=2$.  From Equation~\ref{eqn:thm2-a}, we see $i^*$ will be active at $t=2$ before attenuation with probability $\alp_2=\gam_1 \cdot (1-q_1)>1-1/T$. Thus, $\beta_2=(1-1/T)/\alp_2<1$ and $\gam_2=1-1/T$. Continuing this analysis, we see for each $t=2,3,\ldots, K$, $\alp_t>(1-1/T)^{t-1}$, $\beta_t<1$, and $\gam_t=(1-1/T)^{t-1}$. Now consider the case $t=K+1$. We see that 
\[
\alp_{K+1}=\gam_K \cdot (1-q_K)=(1-1/T)^{K-1} \cdot (1-q_K) \le (1-1/T)^{K}. 
\]
Therefore, $\beta_{K+1}=1$ and $\gam_{K+1}=\alp_{K+1}$. Following this analysis, we have $\alp_t=\gam_t \le (1-1/T)^{t-1}$ and $\beta_t=1$ for all $t \ge K+1$.
\end{proof}

The above lemma implies that we will keep adding a proper attenuation factor $\beta_t<1$ to the agent $i^*$ for all $1<t \le K$, and afterwards, we will essentially add no attenuation to $i^*$. Let $Z=1$ indicate that $i^*$ is matched in the end in \sab.

\begin{theorem}
$\E[Z] \ge 0.719 \cdot x_{i^*}$
\end{theorem}

\begin{proof}
Let $Z=Z_a+Z_b$, where $Z_a=1$ and $Z_b=1$ indicate that $i^*$ is matched during any round $t < K$ and $t \ge K$, respectively. Let $K/T=\kap+o(1)$, where $\kap \in [0,1]$ is a constant and $o(1)$ is a vanishing term when $T \rightarrow \infty$. Let $f(p,x)=\frac{x}{x+(1-x) \cdot p}$. We can verify that for any fixed $p \in (0,1]$, $f(p,x)$ is an increasing concave function over $x \in [0,1]$. 

\xhdr{Lower bounding the value of $\E[Z_a]$}. For each $t < K$, let $Z_t=1$ indicate that $i^*$ is matched during the round of $t$. Observe that $Z_t=1$ iff $i^*$ is active at $t$ after attenuation (\ie $\bo_t=1$) and $i^*$ is inactive at $t+1$ before attenuation (\ie $\bo'_{t+1}=0$). Thus, we have 
\[
\E[Z_t] =\Pr[ (\bo_{t}=1) \wedge (\bo'_{t+1}=0)]=\Pr[\bo_{t}=1] \cdot \Pr[\bo'_{t+1}=0| \bo_{t}=1]=\gam_t \cdot q_t.
\]
Observe that from Lemma~\ref{lem:2-c}, we have $\gam_t=(1-1/T)^{t-1}$ for all
$2<t \le K$ and it is valid for $t=1$ as well. Therefore, we have
\begin{align}
\E[Z_a]&=\sum_{1\le t< K} \E[Z_t]=\sum_{1 \le t \le K} \gam_t \cdot q_t =\sum_{1\le t < K} \bp{1-\frac{1}{T}}^{t-1} \cdot q_t \\
& \ge \sum_{1 \le t < K} \bp{1-\frac{1}{T}}^{t-1} \cdot \frac{1}{T}\sum_{j \sim i^*} \frac{x_{i^*,j}}{x_{i^*,j}+(1-x_{i^*,j})\cdot (1-1/T)^{t-1}} \label{eqn:3-2}.
\end{align}


Recall that  $f(p,x)=\frac{x}{x+(1-x) \cdot p}$ is an increasing concave function over $x \in [0,1]$.  Define $\cS(x)=\{x\}$ if $0 \le x \le 1-\sfe^{-1}$, and $\cS(x)=\{1-\sfe^{-1}, x-(1-\sfe^{-1})\}$ if $1-\sfe^{-1}<x \le 1-\sfe^{-2}$, and $\cS(x)=\{1-\sfe^{-1}, \sfe^{-1}-\sfe^{-2}, x-(1-\sfe^{-2})\}$ if $1-\sfe^{-2}<x \le 1$. Following the same analysis as shown in Lemma~\ref{lem:1-b}, we see that
 \begin{equation}\label{eqn:3-1}
 \sum_{j \sim i^*} \frac{x_{i^*,j}}{x_{i^*,j}+(1-x_{i^*,j})\cdot (1-1/T)^{t-1}} \ge 
 \sum_{x \in \cS(x_{i^*})} f\bp{ (1-1/T)^{t-1}, x}.
 \end{equation}
 
Recall that $K/T=\kap+o(1)$. Plugging  Inequality~\eqref{eqn:3-1} to  Inequality~\eqref{eqn:3-2}, we have
 \begin{align*}
 \E[Z_a] &\ge  \sum_{1 \le t < K} \bp{1-\frac{1}{T}}^{t-1} \cdot \frac{1}{T}\cdot \sum_{x \in \cS(x_{i^*})} f\bp{ (1-1/T)^{t-1}, x}\\
 &=\sum_{x \in \cS(x_{i^*})}  \int_0^{\kap} d\zeta \cdot \sfe^{-\zeta}\cdot f(\sfe^{-\zeta},x). ~~~~\mbox{~~(Taking $T\rightarrow \infty$)}\\
 \end{align*}
\xhdr{Lower bounding the value of $\E[Z_b]$}. By definition, $\E[Z_b]=\sum_{K \le t \le T}\E[Z_t]$. Observe that $i^*$ will be active at (the beginning of) $t$ after attenuation with probability $\E[\bo_{K}]=(1-1/T)^{K-1}$. What's more, there will be no attenuation in essence during all $t>K$. Thus, assume $i^*$ is active after attenuation at $t=K$,  we can apply almost the same analysis as in Section~\ref{sec:thm1} to lower bound $\E[Z_b]$. 

\begingroup
\allowdisplaybreaks
\begin{align*}
\E[Z_b] &=\sum_{t \ge K}\E[Z_t] =\bp{1-\frac{1}{T}}^{K-1}\bp{1-\prod_{t \ge K}(1-q_t)} 
\\
&\ge \bp{1-\frac{1}{T}}^{K-1}\bb{1-\exp \bp{- \sum_{t\ge K}\sum_{j\sim i^*} \frac{1}{T} \cdot \frac{x_{i^*,j}}{x_{i^*,j}+(1-x_{i^*,j}) (1-1/T)^{t-1}}  }}\\
& \ge \bp{1-\frac{1}{T}}^{K-1}\bb{1-\exp \bp{- \sum_{t\ge K} \frac{1}{T}\cdot\sum_{x \in \cS(x_{i^*})} f\bp{ (1-1/T)^{t-1}, x} }}\\
&= \sfe^{-\kap}\bb{1-\exp \bp{- \sum_{x\sim \cS(x_{i^*})} \int_\kap^1 d\zeta \cdot  f\bp{ \sfe^{-\zeta}, x} }   }~~\mbox{(Taking $T\rightarrow \infty$)} \\
&= \sfe^{-\kap}\bb{1-\prod_{x\sim \cS(x_{i^*})}  \exp \bp{- \int_\kap^1 d\zeta \cdot  f\bp{ \sfe^{-\zeta}, x} }   }.
\end{align*}
\endgroup

Putting all the above stuff together, we have 
\[
\E[Z]  \ge F(x_{i^*}, \kap) \doteq  \sum_{x \in \cS(x_{i^*})}  \int_0^{\kap} d\zeta \cdot \sfe^{-\zeta}\cdot f(\sfe^{-\zeta},x) +\sfe^{-\kap}\bb{1-\prod_{x\sim \cS(x_{i^*})}  \exp \bp{- \int_\kap^1 d\zeta \cdot  f\bp{ \sfe^{-\zeta}, x} }   }.
\]
We can verify  via Mathematica that  $\min_{0 \le x_{i^*} \le 1, 0\le \kap \le 1 }F(x_{i^*}, \kap)/x_{i^*} \ge 0.719$ and the inequality becomes tight when $x_{i^*}=1-\sfe^{-1}$ and $\kap=1$.
\end{proof}

%% file: 6-hard.tex
\section{Proof of main Theorems~\ref{thm:main-0} and~\ref{thm:main-4}}\label{sec:thm1-4}

\subsection{Proof of Theorem~\ref{thm:main-0}}
Let us briefly describe \gre and \rank here for  \ifm and \gfm. For \gre,  it will always assign an online arriving agent to an offline available neighbor such that the match can improve the current objective of \ifm and \gfm most; break the tie in a \emph{uniformly} way if there is. For \rank, it will first choose a random permutation $\pi$ over all offline neighbors and then it  will always assign an online arriving agent to an offline available neighbor with the lowest rank in $\pi$. Observe that $\ifm$ is a special case of \gfm when each group consists of one single offline type. Thus, it will suffice to show that \gre and \rank achieve a ratio of zero for \ifm to prove Theorem~\ref{thm:main-0}.

\begin{example}\label{exam:worst}
Consider such an instance $\cI$ of \ifm as follows. We have $|I|=|J|=T=n$ offline and online agents. For $j=1$, it can serve all offline agents, \ie $\cN_{j=1}=I$. For each online agent $j=2,3,\ldots, n$, it can serve one single offline agent $i=j$. Consider such an offline algorithm $\ALG$ (not necessarily a clairvoyant optimal): Try to match each online agent $j \in J$ with $i=j$ if agent $j$ arrives at least once. We can verify that in \alg, each offline agent will be matched with probability at least $1-\sfe^{-1}$. Thus, we claim that for the offline optimal, its performance should have $\OPT \ge \ALG \ge 1-\sfe^{-1}$.  
\end{example}

\begin{lemma}
\gre achieves an online competitive ratio of zero for \ifm Example~\ref{exam:worst}.
\end{lemma}
 
 \begin{proof}
Let $K_t$ be number of unmatched offline agents excluding $i=1$ at time $t$. According to \gre, when $j=1$ arrives at $t$, it will match $i=1$ with a probability $1/(K_t+1)$ if $i=1$ is not matched then. For each $i\neq 1$, let $Z_{i,t}=1$ indicate that $i$ is matched by time $t$. Thus, we have $K_t=(n-1)-\sum_{i=2}^{n} Z_{i,t}$. For each $j \in J$, let $A_{j,t}$ be the number of arrivals of agent $j$ by $t$, and let  $\bo_{j,t}=\min(1, A_{j,t})$ indicate if $j$ arrives at least once by $t$. Thus, $\sum_{i=2}^n Z_{i,t} \le \sum_{j=2}^n \bo_{j,t}+A_{1,t}$. Observe that $A_{1,t}$ can be viewed as the sum of $t$ \iid Bernoulli random variables each with $1/T$. By Chernoff bound, we have  $\Pr[A_{1,t} \ge (n-1)/(4\sfe)]\le \sfe^{-\Omega(T^3/t)}$ (note that $n=T$).
\begin{align}
&\Pr\bb{K_t \le \frac{1}{2}\cdot (n-1) \cdot \sfe^{-t/T-o(t/T)}}=\Pr\bb{(n-1)-\sum_{i=2}^{n} Z_{i,t}\le  \frac{1}{2}\cdot (n-1) \cdot \sfe^{-t/T-o(t/T)}}\\
&=\Pr\bb{\sum_{i=2}^{n} Z_{i,t}\ge n-1-  \frac{1}{2}\cdot (n-1) \cdot \sfe^{-t/T-o(t/T)}}\\
&\le
\Pr\bb{\sum_{j=2}^n \bo_{j,t}+A_{1,t}\ge n-1-  \frac{1}{2}\cdot (n-1) \cdot \sfe^{-t/T-o(t/T)}}\\
&\le \sfe^{-\Omega(T^3/t)}+\Pr\bb{\sum_{j=2}^n \bo_{j,t}+(n-1)/(4\sfe) \ge n-1-  \frac{1}{2}\cdot (n-1) \cdot \sfe^{-t/T-o(t/T)}} \\
&= \sfe^{-\Omega(T^3/t)}+\Pr\bb{\sum_{j=2}^n \bo_{j,t}-\E\big[\sum_{j=2}^n\bo_{j,t}\big] \ge \frac{n-1}{2} \cdot \sfe^{-t/T-o(t/T)}-\frac{n-1}{4\sfe}}\label{ineq:rg-a}\\
& \le  \sfe^{-\Omega(T^3/t)}+\Pr\bb{\sum_{j=2}^n \bo_{j,t}-\E\big[\sum_{j=2}^n\bo_{j,t}\big] \ge \frac{n-1}{4} \cdot \sfe^{-t/T-o(t/T)}}\label{ineq:rg-b}\\
&\le  \sfe^{-\Omega(T^3/t)}+ \exp\bp{-\frac{n-1}{8} \cdot \sfe^{-2t/T-o(t/T)}}\label{ineq:rg-c}.
\end{align}
Note that Inequality~\eqref{ineq:rg-a} is due to that $\E\big[\sum_{j=2}^n\bo_{j,t}\big]=(n-1)\cdot (1-1/T)^{t-1}=(n-1)\cdot (1- \sfe^{-t/T-o(t/T)})$; Inequality~\eqref{ineq:rg-b} is due to $\frac{1}{4\sfe} \le \frac{1}{2} \cdot\frac{1}{2} \cdot \sfe^{-t/T-o(t/T)}$; Inequality~\eqref{ineq:rg-c} is due to Chernoff-Hoeffding bound.

Recall that $Z_{i=1}=1$ indicates that $i=1$ is matched by the end of $T$. We have that 
 \begin{align*}
 \E[Z_{i=1}]& \le \frac{1}{T}\sum_{t=1}^{T} \bP{\sfe^{-\Omega(T^3/t)}+\exp\bp{-\frac{n-1}{8} \cdot \sfe^{-2t/T-o(t/T)}} + \frac{1}{\frac{1}{2}\cdot (n-1) \cdot\sfe^{-t/T-o(t/T)}+1}}\\
&=O(\sfe^{-\Omega(T^2)})+\int_0^1 d\zeta \exp\bp{-\frac{n-1}{8} \cdot \sfe^{-2 \zeta}}+\int_0^1 d\zeta \frac{1}{\frac{1}{2}\cdot (n-1)\cdot\sfe^{-\zeta}+1} \\
&\le O(\sfe^{-\Omega(T^2)})+ \exp\bp{-\frac{n-1}{8 \sfe^2}}+\frac{2(\sfe-1)}{n-1}.
 \end{align*}
This is in contrast with that $\OPT \ge 1-1/\sfe$. Thus, we claim that \gre achieves a ratio of zero.
\end{proof}

\begin{lemma}\label{lem:ranking}
\rank has a competive ratio of zero for \ifm on Example~\ref{exam:worst}.
\end{lemma}

\begin{proof}
Let $\cS$ be the (random) set of indices of offline nodes that fall before $i=1$ under $\pi$. Consider a given $\cS$ with $|\cS|=K$. For each $j \in J$, let $A_{j,t}$ be the number of arrivals of online agent $j \in J$ before the start of time $t \in [T]$. For each $i \in \cS$,  let $Z_{i,t}=1$ indicate that $i \in \cS$ is matched by $t$. Observe that 
\[
\Pr\bb{\sum_{i \in \cS} Z_{i,t} \ge K} \le \Pr\bb{\sum_{j \in \cS} \min(A_{j,t},1)+ A_{1,t} \ge K}.
\]

Observe that $A_{1,t}$ can be viewed as the sum of $t$ \iid Bernoulli random variables each with $1/T$. By Chernoff bound, $\Pr[A_{1,t} \ge K/\sfe^2] \le \sfe^{-\Omega(K^2 \cdot T/t)}$. Thus, we have that
\begin{align}
\Pr\bb{\sum_{i \in \cS} Z_{i,t} \ge K} &\le \Pr\bb{\sum_{j \in \cS} \min(A_{j,t},1)+ A_{1,t} \ge K} \\
&\le \exp\bp{-\Omega(K^2 \cdot T/t)}+\Pr\bb{\sum_{j \in \cS} \min(A_{j,t},1)+K/\sfe^2 \ge K} \label{ineq:rank-1}.
\end{align}
Observe that (1) $\{\min(A_{j,t},1)\}$ are negatively associated due to \cite{joag1983negative}; (2) 
$\E[\min(A_{j,t},1)]=1-(1-1/T)^t=1-\sfe^{-t/T-o(t/T)}$ for each $j \in \cS$. By applying Chernoff-Hoeffding bound, we have
\begin{align}
&\Pr\bb{\sum_{j \in \cS} \min(A_{j,t},1)+K/\sfe^2 \ge K}\\
&=\Pr\bb{\sum_{j \in \cS} \min(A_{j,t},1)-\E\big[\sum_{j \in \cS} \min(A_{j,t},1)\big] \ge K-K/\sfe^2-
K \cdot \bp{1-\sfe^{-t/T-o(t/T)}}}\\
&= \Pr\bb{\sum_{j \in \cS} \min(A_{j,t},1)-\E\big[\sum_{j \in \cS} \min(A_{j,t},1)\big] \ge K\cdot \sfe^{-t/T-o(t/T)}-K/\sfe^2} \\
&\le \Pr\bb{\sum_{j \in \cS} \min(A_{j,t},1)-\E\big[\sum_{j \in \cS} \min(A_{j,t},1)\big] \ge K\cdot \sfe^{-t/T-o(t/T)}/2}\\
&\le \exp\bp{-\sfe^{-2t/T-o(t/T)}\cdot K/2}.
\end{align}
Thus, plugging into the above result to Inequality~\eqref{ineq:rank-1}, we have
\[
\Pr[\sum_{i \in \cS} Z_{i,t} \ge K] \le  \exp\bp{-\Omega(K^2 \cdot T/t)}+ \exp\bp{-\sfe^{-2t/T-o(t/T)}\cdot K/2}.
\]
Consider a given $\cS$ with $|\cS|=K$. We see that
\begin{align*}
\E[Z_1|K]& \le \frac{1}{T}\sum_{t=1}^T \Pr\bb{\sum_{i \in \cS} Z_{i,t} \ge K}\\
& \le \frac{1}{T}\sum_{t=1}^T \bB{ \exp\bp{-\Omega(K^2 \cdot T/t)}+ \exp\bp{-\sfe^{-2t/T-o(t/T)}\cdot K/2}}\\
&=\sfe^{-\Omega(K^2)}+\int_0^1 \sfe^{-(K/2) \cdot \sfe^{-2 \zeta}} d\zeta \le \sfe^{-\Omega(K^2)}+ \sfe^{-  K / (2\sfe^2)}. 
\end{align*}
Observe that $K$ takes values $0,1,2,\ldots, n-1$ with a uniform probability $1/n$. Thus,
\[
\E[Z_1]\le \frac{1}{n} \sum_{K=0}^{n-1} \bp{\sfe^{-\Omega(K^2)}+ \sfe^{-K / (2\sfe^2)}}=O(1/n).
\]
This is in contrast with that $\OPT \ge 1-1/\sfe$. Thus, we claim that \rank achieves a ratio of zero.
\end{proof}

\subsection{Proof of Theorem~\ref{thm:main-4}}

Consider any bipartite graph $(I,J,E)$ and let $T$ denote the length of the time horizon.
Let $\Psi$ denote the \textit{finite} set of all \textit{deterministic} online matching policies given the graph and $T$.
For any $\psi\in\Psi$ and offline node $i\in I$, let $q_{i,\psi}$ denote the probability (over the random arrival draws) that $i$ gets matched under algorithm $\psi$ by the end of the time horizon.

Let $n=|I|$ and fix a feasible solution $x\in[0,1]^n$ to the LP~\eqref{cons:j}--\eqref{cons:e} for the graph.
Consider the following LP:
\begin{align*}
\max\ \gamma
\\ \text{s.t. }\sum_{\psi\in\Psi}q_{i,\psi}z_{\psi} &\ge x_i\gamma &\forall i=1,\ldots,n
\\ \sum_{\psi\in\Psi}z_{\psi} &=1
\\ z_{\psi} &\ge0 &\forall\psi\in\Psi
\end{align*}
where variable $z_{\phi}$ represents the probability that a randomized algorithm for \ifm or \gfm selects deterministic policy $\psi$.
Objective $\gamma$ is set to the maximum value for which the randomized online algorithm can uniformly guarantee a matching probability of $x_i\gamma$ for every offline agent $i$.
Taking the dual of this LP, we get:
\begin{align}
\min\ \theta \nonumber
\\ \text{s.t. }\sum_{i=1}^nq_{i,\psi}w_i &\le\theta &\forall\psi\in\Psi \label{eqn:thetaEqn}
\\ \sum_{i=1}^nx_iw_i &=1 \nonumber
\\ w_i &\ge0 &\forall i=1,\ldots,n \nonumber
\end{align}

Since $\Psi$ is finite, by strong LP duality, whenever there exists a $x\in[0,1]^n$ such that the optimal objective value of the primal LP is $\gamma=c$, there exist feasible weights $w_i\ge0$ such that in the dual LP,~\eqref{eqn:thetaEqn} holds with $\theta=c$.
That is, the LP for \vw has a feasible solution $x$ with objective value $\sum_ix_iw_i=1$, yet any deterministic online policy $\psi$ cannot earn than $c$ (by~\eqref{eqn:thetaEqn}).
Since deterministic online policies are optimal in \vw with known weights, this shows that the competitive ratio for \vw cannot be better than $c$.

To complete the proof, consider any upper bound $c$ on the competitive ratio of \gfm.
There must exist a $x\in[0,1]^n$ such that the primal objective value is no greater than $c$ (otherwise it would be possible to achieve a competitive ratio of $c$ in \gfm).
By the preceding argument, such an $x$ would also upper-bound the competitive ratio of \vw by $c$.
Taking an infimum over all valid upper bounds $c$, we complete the proof that $\Phi(\vw)\le\Phi(\gfm)$.  In turn, $\Phi(\gfm)\le\Phi(\ifm)$ since $\ifm$ is a special case of $\gfm$.
Note that the converse inequality of $\Phi(\vw)\ge\Phi(\gfm)$ cannot be argued in the same way since the groups may overlap arbitrarily; the converse inequality of $\Phi(\vw)\ge\Phi(\ifm)$ also cannot be argued since the $x\in[0,1]^n$ for \vw may not satisfy $x_1=\cdots=x_n$.

%% file: 6-exp.tex
\section{Experimental results}\label{sec:exp}

\subsection{Experiments on \ifm and \gfm}
\xhdr{Preprocessing of a Ride-hailing dataset}.
We test our algorithms of \ifm and \gfm on a public ride-hailing dataset, which is collected  from the city of Chicago\footnote{https://data.cityofchicago.org/Transportation/Transportation-Network-Providers-Trips/m6dm-c72p}. Following the setting in~\cite{nanda2020balancing,xu2020trade}, we focus on a short time window and assume that drivers are offline agents  while riders are online agents that arrive dynamically. Our goal is to maximize individual and group fairness among all drivers. The dataset has more than $169$ million trips starting from November 2018. Each trip record includes the trip length, the starting and ending time, the pick-up and drop-off locations for the passenger, and some other information such as the fare and the tip. Note that Chicago is made up of $77$ community areas that are well defined and do not overlap. Thus, we can categorize all trips according to the pre-defined community areas. In our case, we define a driver group for each of the $77$ areas and assume each driver belongs to the group identified as her starting community area. Recall that our metric of group fairness is defined as the minimum matching rates of offline agents over all groups, which affects the minimum average-earning-rate among ride-hailing drivers across different community areas in Chicago. Since locations reflect the racial- and socioeconomic-disparities in Chicago\footnote{https://statisticalatlas.com/place/Illinois/Chicago/Race-and-Ethnicity}, we believe our objective of maximizing group fairness among drivers across different locations can help promote racial and social equity as well.

We construct the input bipartite graph as follows. We focus on the time window from $18:00$ to $19:00$ on September 29, 2020, and subsample $T$ trips from a total of $11,228$ trips. For each trip, we create an individual driver $i$ and rider $j$, where $i$ has an attribute of a starting community area while $j$ has an attribute of a pair of starting and ending areas. In this way, we have $|I|=|J|=T$. Note that it is possible for multiple drivers to share the same starting community area. In this case, we assume they belong to the same group identified by the starting area. For each driver-rider pair, we add an edge if they share the same starting area. 

\xhdr{Algorithms}. For the problems \ifm and \gfm, we compare our algorithm \sbo against the following: (a) \textbf{\alggreedy}: Assign each arriving agent to an available neighbor (\ifm) or an available neighbor whose group has the lowest matching rate at the time of arrival (\gfm); break ties uniformly at random. (b) \textbf{\algranking}: Fix a uniform random permutation of $I$ at the start; assign each arriving agent to the adjacent available offline agent who is earliest in this order. (c) \algbru: the algorithm from \cite{brubach2017} but customized to our setting by replacing its benchmark-LP objectives with  $\LP~\eqref{obj-1}$ (\ifm) and  $\LP~\eqref{obj-2}$ (\gfm). (d) \algmgs: similar to our boosting algorithm, but following \cite{bib:Manshadi}, generating two random candidate neighbors upon the arrival of every online agent instead, and matching it with the first available one. Note that here we do not compare against the algorithm in~\cite{bib:Jaillet}, which relies on the special structure of the LP solution that all $x_e \in \{0,1/3,2/3\}$. This structure unfortunately no longer holds when the objective is either $\LP~\eqref{obj-1}$ (\ifm) or  $\LP~\eqref{obj-2}$ (\gfm).

\begin{figure}[t!]
  \centering
  \begin{subfigure}[b]{0.43\linewidth}
    \includegraphics[width=\linewidth]{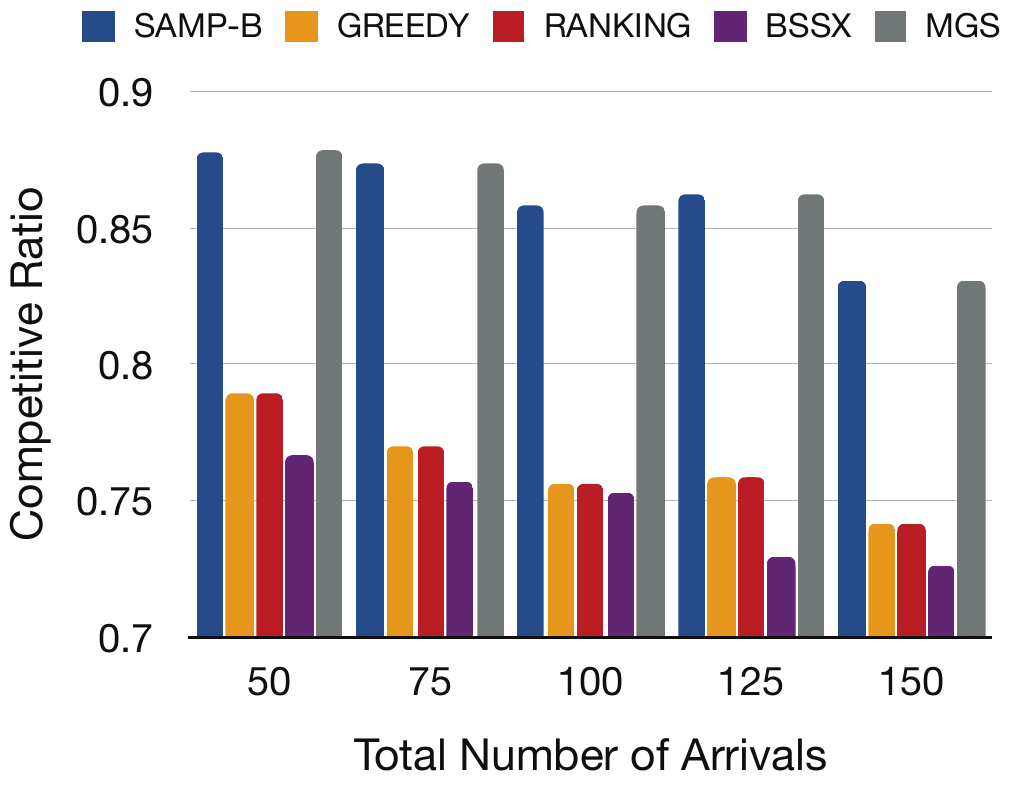}
    \caption{Experimental results on \ifm.}
    \label{fig:ifm}
  \end{subfigure}
  \begin{subfigure}[b]{0.43\linewidth}
    \includegraphics[width=\linewidth]{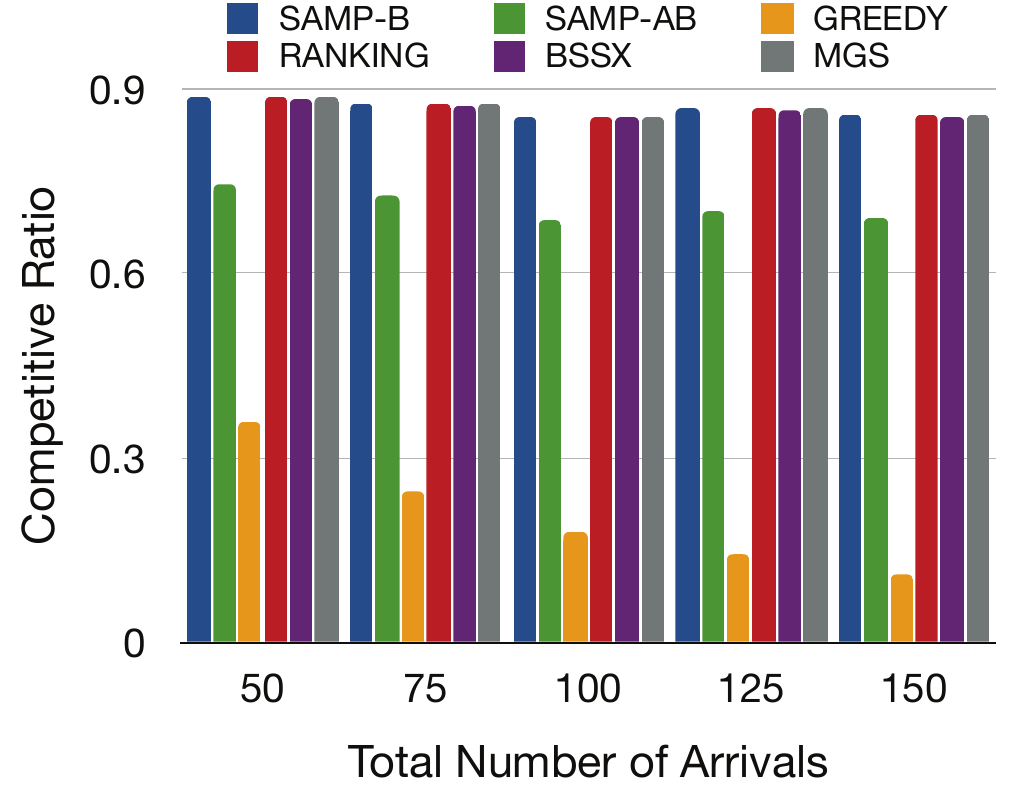}
    \caption{Experimental results on \gfm.}
    \label{fig:gfm}
  \end{subfigure}
  \caption{Experimental results of \ifm and \gfm on a real ride-hailing dataset in Chicago: The total number of arrivals $T$ takes values from $\{50, 75, 100, 125, 150\}$ with $|I| = |J| = T$.}
  \label{fig:result_fm}
\end{figure}

\xhdr{Results and discussions}.
For the real dataset, we vary the number of subsampled trips $T$ in $\{50, 75, 100, 125,\\ 150\}$.
We first construct $100$ subsampled instances for each given $T$, and then run $100$ trials on each instance, reporting the average performance. Note that in the offline phase of \sab,  when it comes to  estimation of the attenuation factor $\beta_{i,t}$ for $i$ at $t$, we  apply the Monte-Carlo method by simulating $100$ times and then taking the average.

Figure~\ref{fig:ifm} shows that for \ifm,  \sbo performs as well as~\algmgs, and both have a significant advantage over \alggreedy, \algranking, and \algbru. The competitive ratios of \sbo always stay above $0.722$, which is consistent with our theoretical bound in Theorem~\ref{thm:main-1}. Figure~\ref{fig:gfm} shows that for \gfm, \sbo performs as good as  \algranking, \algbru and \algmgs, and only $\sab$ and $\alggreedy$ fall behind. That being said, unlike \alggreedy, \sab achieves a steady ratio well above $0.7$ over different choices of $T$. This is consistent with results in~Theorem~\ref{thm:main-2}. All results here suggest that \sbo and \algmgs are top two candidates in practice for both \ifm and \gfm.

We emphasize that although our \sbo algorithm does not significantly outperform the (fairness-adapted) \algmgs algorithm from the literature, it is both conceptually and implementation-wise much simpler. To our understanding, it is surprising that such a simple adaptive boosting algorithm has not been analyzed and extensively tested before.

\subsection{Experiments on \vw}
\xhdr{Construction of input instances}. We acknowledge that it is hard to identify real applications that can perfectly fit the model of \vw.~\cite{borodin2020} has conducted comprehensive experimental studies, which compare the performance of different algorithms for unweighted online matching under KIID on a wide variety of synthetic and real datasets. They proposed an idea, called  \emph{random balanced partition method}, to generate a bipartite graph from a practical social network.  The details are as follows. Suppose we have a real social network with $V$ being the set of vertices and $E$ being the set of edges.  The method partitions $V$ uniformly randomly into two blocks $L$ and $R$, such that $|L|=\lfloor |V| / 2 \rfloor$ and $|R| = \lceil |V| / 2 \rceil$. It keeps only those edges that connect two vertices from the two different partitions. As indicated by~\cite{borodin2020}, research on how to form a maximum matching on a bipartite graph constructed from a real social network can offer great insights regarding how to boost friendship ties among users active in online social platforms (\eg Facebook). 

We follow the idea in~\cite{borodin2020} and select four datasets from the Network Data Repository~\cite{DBLP:conf/aaai/RossiA15}, namely, \textbf{socfb-Caltech36}, \textbf{socfb-Reed98}, \textbf{econ-because} and \textbf{econ-mbeaflw}. The former two datasets are Facebook social-network graphs, where vertices are users and edges are friendship ties. The latter two datasets are two economic networks collected from the U.S.A. in 1972, where vertices are commodities/industries and edges are economic transactions. We list detailed statistics of these $4$ datasets in Table~\ref{tab:vom-real-datasets}. For each network graph $(V,E)$, we first downsample the network size $|V|$ to $200$. Since the original graphs are non-bipartite, we first partition all nodes uniformly at random into two blocks to construct $I$ and $J$, such that $|I| = \lfloor |V| / 2 \rfloor$ and $|J| = \lceil |V| / 2 \rceil$. We keep only the edges that connect two vertices from different partitions. We assign the weight for each offline vertex $i$ to be a random value, uniformly selected from $[0,1]$.
\begin{table}[tb!]
\centering
\caption{Network data statistics for graphs from the Network Data Repository~\cite{DBLP:conf/aaai/RossiA15}.}
\begin{tabular}{c|c|c|c|c|c}
\hline
\textbf{}                & \textbf{Nodes} & \textbf{Edges} & \textbf{Max degree} & \textbf{Min degree} & \textbf{Ave. degree} \\ \hline
\textbf{socfb-Caltech36} & 769            & 16,700         & 248                     & 1                       & 43                      \\ \hline
\textbf{socfb-Reed98}    & 962            & 18,800         & 313                     & 1                       & 39                      \\ \hline
\textbf{econ-beause}     & 507            & 44,200         & 766                     & 2                       & 174                     \\ \hline
\textbf{econ-mbeaflw}    & 492            & 49,500         & 679                     & 0                       & 201                     \\ \hline
\end{tabular}

\label{tab:vom-real-datasets}
\end{table}

\xhdr{Algorithms}. Similar to \ifm and \gfm, we compare  \sbo and \sab against several baselines, including $\alggreedy$, $\algranking$, \algbru~\cite{brubach2017}, and \algmgs~\cite{bib:Manshadi}. Additionally, we test the algorithm presented by~\cite{bib:Jaillet}, denoted by~\alglu. For each of the four instances, we run the above $7$ algorithms for $100$ times and take the average as the final performance. 

\xhdr{Results and discussion}.
Figure~\ref{fig:vom} shows that \sbo is second only to \alggreedy and is comparable to \alggreedy in half of the total instances. The gap between \sbo and \alggreedy declines as the average degree of all nodes increases; see \textbf{econ-because} and \textbf{econ-mbeaflw}. For all instances, \sbo outperforms the other three LP-based algorithms, \algbru, \algmgs and \alglu, all of which involve a much more complicated implementation. This establishes the superiority of  \sbo  in practical instances of \vw over the three LP-based baselines. We observe that the competitive ratios of \sab are always above $0.719$, which is consistent with our theoretical bound in Theorem~\ref{thm:main-2}. Also, note that  \sab can beat the rest three \LP-based baselines in almost all scenarios (except for \textbf{socfb-Reed98}), which suggests that $\sab$ is a top candidate among all all LP-based algorithms.

\begin{figure}[t!]
  \centering
  {\includegraphics[width=0.5 \columnwidth]{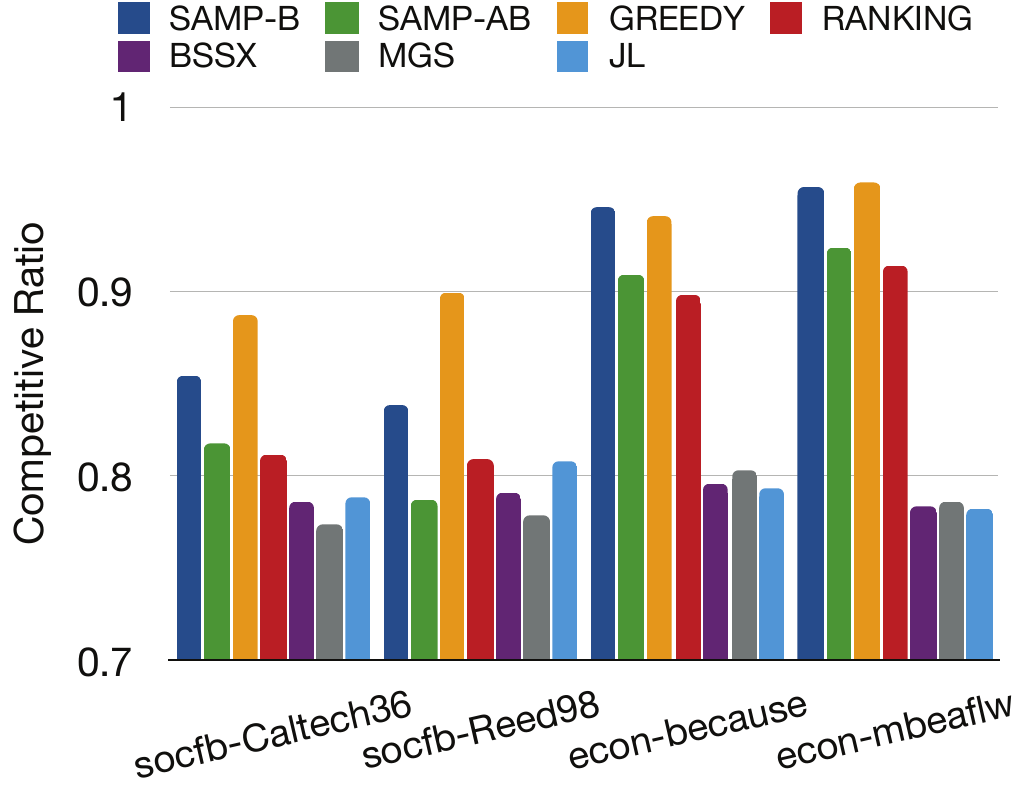}}
     \caption{Experimental results of \vw with on four real datasets from the Network Data Repository~\cite{DBLP:conf/aaai/RossiA15}.}
     \label{fig:vom}
\end{figure}

\subsection{Experiments on \ifm with synthetic datasets}

\xhdr{Preprocessing of a synthetic dataset}.
We first construct the bipartite graph $G=(I,J,E)$, where we always fix $|I| = 100$ and set $|J| = T$. For each offline agent $i$, we randomly choose $\delta$ online agents among $J$ as its neighbors, \ie $|\cN_i|=\delta, \forall i \in I$. In this way, the edge set $E$ is constructed. 
Recall that we consider integral arrival rates for all offline agents, thus, we assume \wl that all $r_j = 1$.

\xhdr{Algorithms}.
We compare our algorithm \sbo against the following:
(a) \textbf{\alggreedy}: Assign each arriving agent to an available neighbor at the time of arrival; break ties uniformly at random. (b) \textbf{\algranking}: Fix a uniformly random permutation of $I$ at the start; assign each arriving agent to the adjacent available offline agent who is earliest in this order. (c) \algbru: the algorithm from \cite{brubach2017} but customized to our setting by replacing its benchmark-LP objectives with $\LP~\eqref{obj-1}$.

\begin{figure}[ht!]
  \centering
  \begin{subfigure}[b]{0.43\linewidth}
    \includegraphics[width=\linewidth]{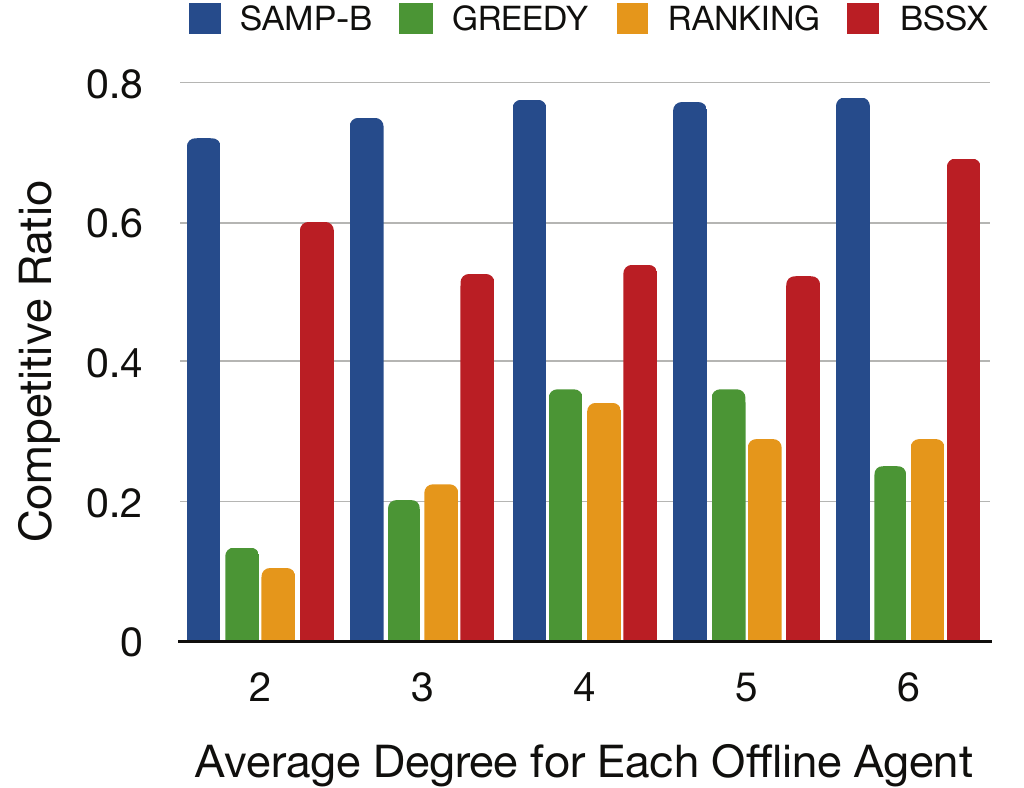}
    \caption{Experimental results when varying average degree $\delta$, while fixing $|I| = 100$ and $|J| = T = 100$.}
    \label{fig:ofm_degree_100}
  \end{subfigure}
  \hspace{5mm}
  \begin{subfigure}[b]{0.43\linewidth}
    \includegraphics[width=\linewidth]{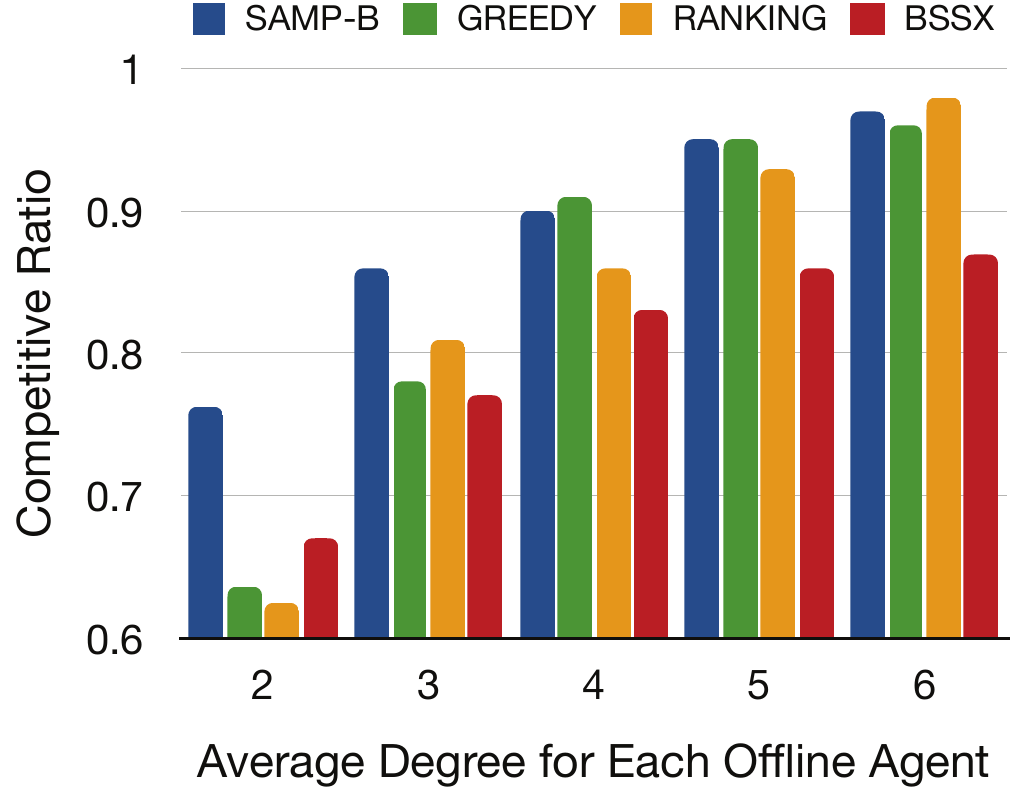}
    \caption{Experimental results when varying average degree $\delta$, while fixing $|I| = 100$ and $|J| = T = 500$.}
    \label{fig:ofm_degree_500}
  \end{subfigure}
  \caption{Experiments on \ifm with synthetic dataset: The average degree $\delta$ takes values from $\{2,3,4,5,6\}$ and the total number of arrivals $T$ takes values from $\{100,500\}$ with $|I| = 100$.}
  \label{fig:ofm_degree}
\end{figure}

\begin{figure}[ht!]
  \centering
  \begin{subfigure}[b]{0.43\linewidth}
    \includegraphics[width=\linewidth]{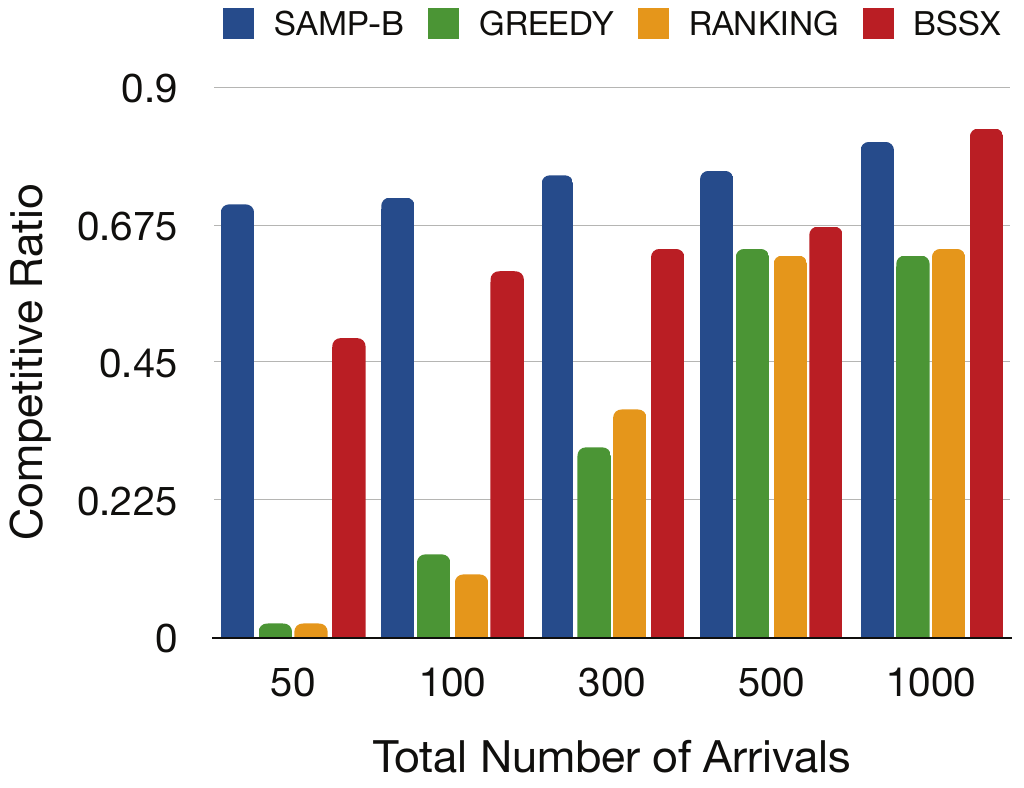}
    \caption{Experimental results when varying $T$, while fixing $|I| = 100$ and $\delta = 2$.}
    \label{fig:ofm_T_d2}
  \end{subfigure}
  \hspace{5mm}
  \begin{subfigure}[b]{0.43\linewidth}
    \includegraphics[width=\linewidth]{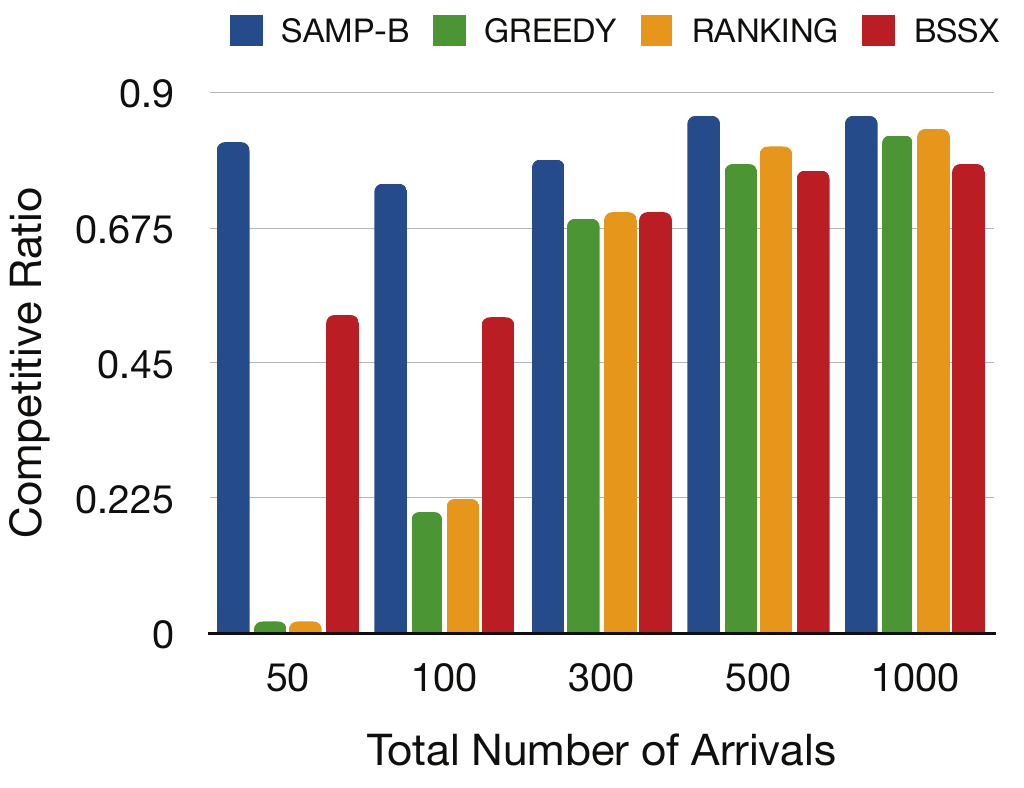}
    \caption{Experimental results when varying $T$, while fixing $|I| = 100$ and $\delta = 3$.}
    \label{fig:ofm_T_d3}
  \end{subfigure}
  \caption{Experiments on \ifm with synthetic dataset: The total number of arrivals $T$ takes values from $\{50, 100, 300, 500, 1000\}$ and the average degree $\delta$ takes values from $\{2,3\}$ with $|I| = 100$.}
  \label{fig:ofm_T}
\end{figure}

\xhdr{Results and discussion}.
For synthetic dataset, we vary the average degree $\delta$ in $\{2,3,4,5,6\}$ and the total number of arrivals $T$ in $\{50,100,300,500,1000\}$, respectively. For each instance, we run all algorithms for $100$ times and take the average as the final performance.

Figure~\ref{fig:ofm_degree} and Figure~\ref{fig:ofm_T} show the results on synthetic dataset. Almost in all cases, \sbo is the clear winner and has the best performance. At times two heuristics do well, as shown in Figure~\ref{fig:ofm_degree_500}. This is due to the fact that the more resources we have (with a larger average degree $\delta$ and total number of arrivals $T$), the less planning we need. However, when we set $|I| = |J| = T$, as shown in Figure~\ref{fig:ofm_degree_100}, \sbo outperforms the two heuristics significantly. In addition, Figure~\ref{fig:ofm_T} shows that when $T$ is small, \sbo is  close to its competitive ratio of $0.722$, as given in Theorem~\ref{thm:main-1}. This highlights the tightness of our theoretical lower bound.

\subsection{Experiments on \vw with synthetic datasets}

\xhdr{Synthetic dataset}. 
Our experimental setup for the synthetic datasets is as follows.
We first construct the bipartite graph $G=(I,J,E)$, where we always set $|I| = |J| = T$. 
For each offline agent $i$, we randomly choose $\delta$ online agents among $J$ as its neighbors, \ie $|\cN_i|=\delta, \forall i \in I$.
Note that $|\cN_j|=\delta, \forall j \in J$ due to $|I| = |J|$. We assign the weight for each offline vertex $i$ to be a random value, uniformly selected from $[0,1]$.

\xhdr{Algorithms}.
Similar to \ifm, we compare \sbo and \sab against several baselines, including $\gre$, $\algranking$, \algbru~\cite{brubach2017}. Additionally, we test the algorithm presented by~\cite{bib:Jaillet}, denoted by~\alglu.
For each instance, we run the above $6$ algorithms for $100$ times and take the average as the final performance. Note that in the offline phase of \sab, we estimate the attenuation factor $\beta_{i,t}$ for $i$ at $t$ by applying Monte-Carlo method for $100$ times. For each algorithm, we compute two kinds of competitive ratio as follows: (1) \cri: $\min_i \mathbb{E}[Z_i]/x_i^*$. (2) \crii: $\sum_i w_i* \mathbb{E}[Z_i]$.

\begin{figure}[ht!]
  \centering
  \begin{subfigure}[b]{0.43\linewidth}
    \includegraphics[width=\linewidth]{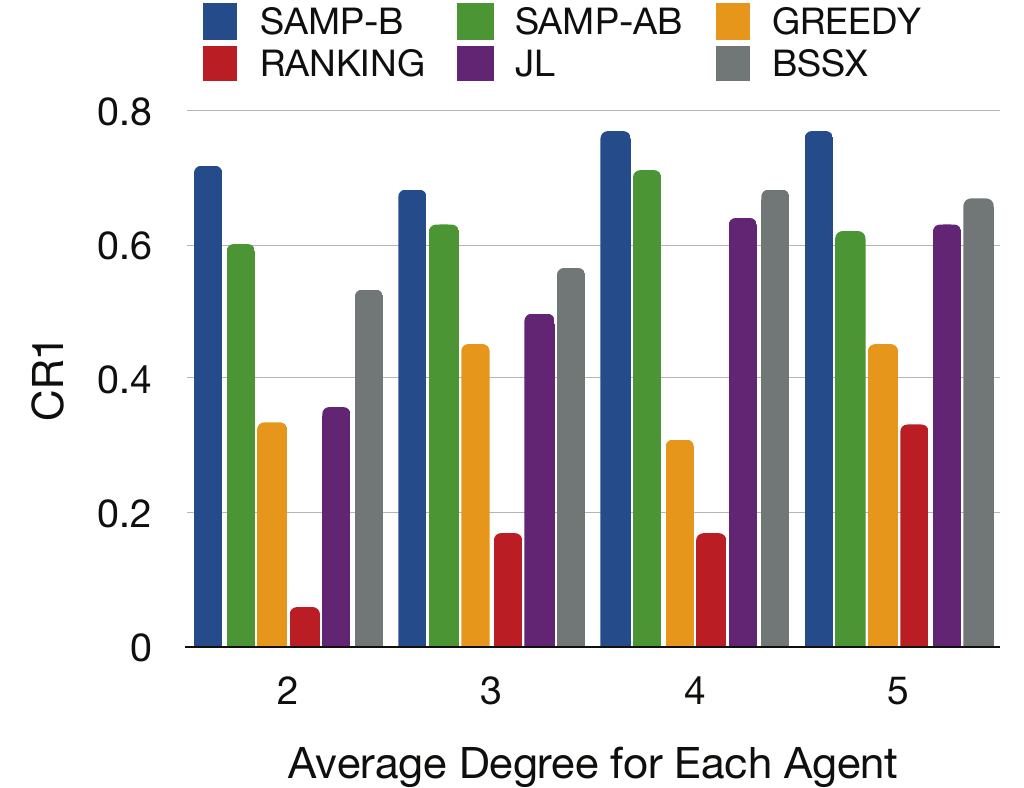}
    \caption{\cri achieved when varying average degree $\delta$.}
    \label{fig:vom_degree_cr1}
  \end{subfigure}
  \hspace{5mm}
  \begin{subfigure}[b]{0.43\linewidth}
    \includegraphics[width=\linewidth]{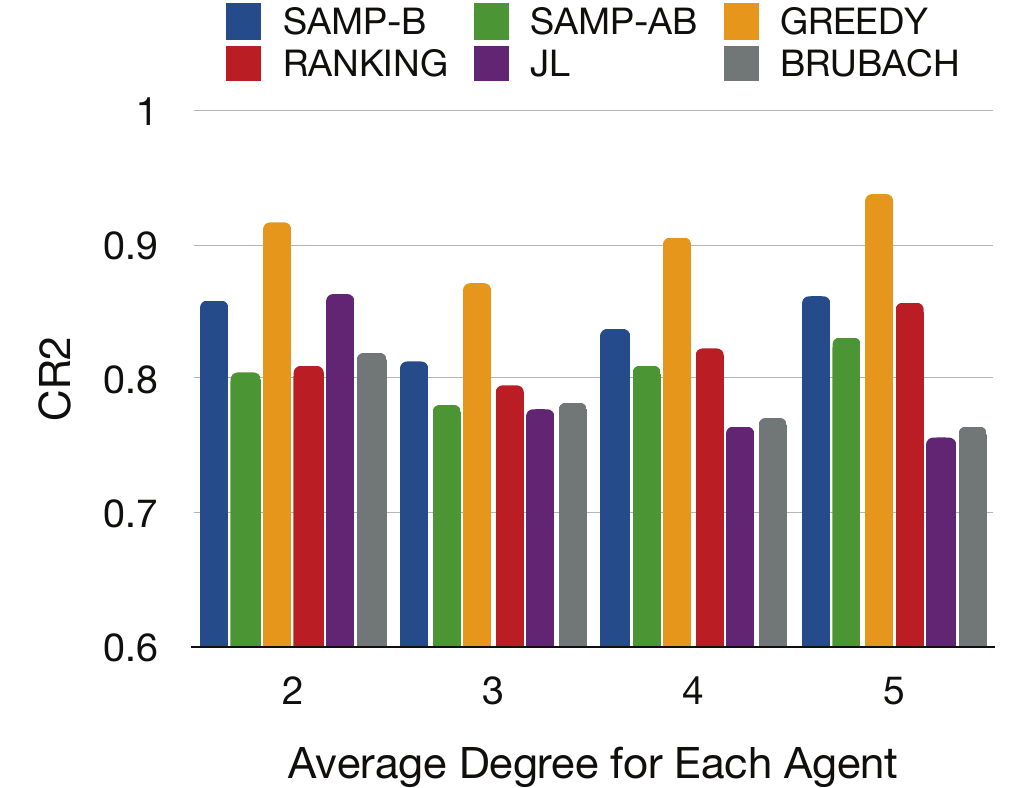}
    \caption{\crii achieved when varying average degree $\delta$.}
    \label{fig:vom_degree_cr2}
  \end{subfigure}
  \caption{Experiments on \vw with synthetic dataset: The average degree $\delta$ takes values from $\{2,\mathbf{3},4,5,6\}$ with $|I| = |J| = T = 100$.}
  \label{fig:vom_degree}
\end{figure}

\begin{figure}[ht!]
  \centering
  \begin{subfigure}[b]{0.43\linewidth}
    \includegraphics[width=\linewidth]{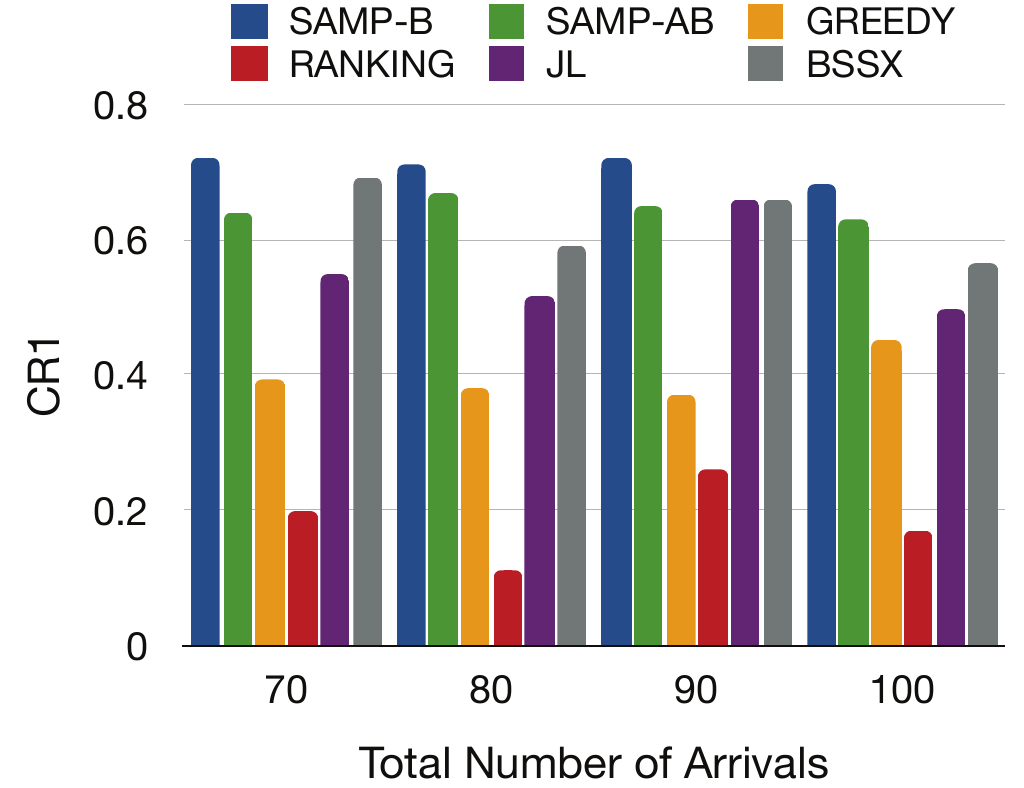}
    \caption{\cri achieved when varying the total number of arrivals $T$.}
    \label{fig:vom_T_cr1}
  \end{subfigure}
  \hspace{5mm}
  \begin{subfigure}[b]{0.43\linewidth}
    \includegraphics[width=\linewidth]{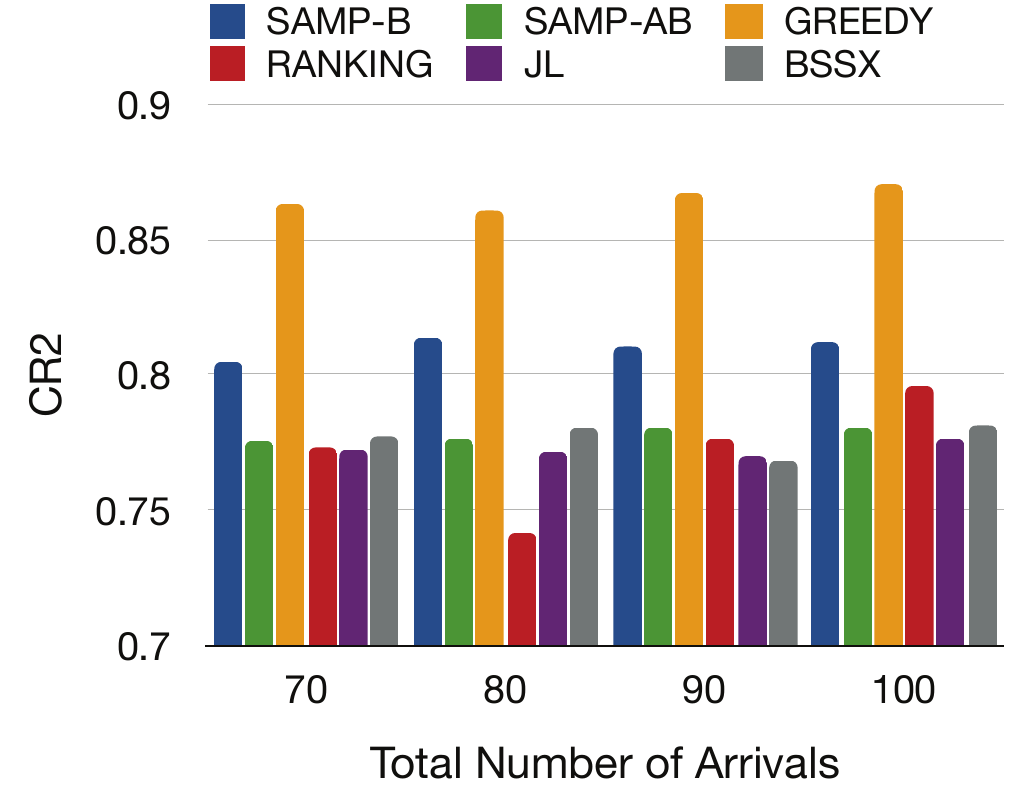}
    \caption{\crii achieved when varying the total number of arrivals $T$.}
    \label{fig:vom_T_cr2}
  \end{subfigure}
  \caption{Experiments on \vw with synthetic dataset: The total number of arrivals $T$ takes values from $\{70,80,90,\mathbf{100}\}$ with the average degree $\delta = 3$.}
  \label{fig:vom_T}
\end{figure}

\xhdr{Results and discussion}.
For synthetic dataset, we run two kinds of experiments by varying the average degree $\delta$ in $\{2,\mathbf{3},4,5,6\}$ and the total number of arrivals $T$ in $\{70,80,90,\mathbf{100}\}$, respectively. Overall, the \cri achieved by each algorithm is strictly dominated by its corresponding \crii. This is expected since \cri takes the minimum expected matching ratio over all $i$ as the final ratio. The higher \cri an algorithm achieves, the more robust the algorithm will be. As shown in Figure~\ref{fig:vom_degree} and \ref{fig:vom_T}, we highlight that \sbo displays remarkable robustness, \ie the best \cri achieved, significantly outperforming the two heuristics, although \sbo has a slightly lower \crii when compared with \alggreedy. Note that in all cases, the \crii achieved by \sbo and \sab is always above $0.719$, which is consistent with our theoretical prediction in Theorem~\ref{thm:main-2}. Another observation is that \sbo can outperform the other two LP-based algorithms, \algbru and \alglu, all of which involve a much more complicated implementation.